\documentclass{vldb}
\usepackage{booktabs} 
\usepackage{graphicx}
\usepackage{here}
\usepackage[ruled, linesnumbered, vlined]{algorithm2e}
\usepackage{subfigure}
\usepackage{calc}
\usepackage{tabularx}
\usepackage{hyperref}
\usepackage{balance}
\usepackage{mathtools}
\usepackage{lipsum,multicol}
\usepackage{soul}
\usepackage{stfloats}

\usepackage{amsmath}
\newcommand{\argmin}{\operatornamewithlimits{argmin}}
\mathchardef\mhyphen="2D 

\usepackage{amsthm}
\newtheorem{theorem}{Theorem}

\theoremstyle{remark}

\usepackage{color, colortbl}

\definecolor{LightCyan}{rgb}{0.88,1,1}

\newenvironment*{highlight}{\color{black}}{}

\usepackage{array,multirow}
\usepackage{float}

\usepackage{listings}

\makeatletter
\def\hlinewd#1{%
\noalign{\ifnum0=`}\fi\hrule \@height #1 %
\futurelet\reserved@a\@xhline}
\makeatother

\makeatletter
\newcommand{\removelatexerror}{\let\@latex@error\@gobble}
\makeatother
\let\oldnl\nl
\newcommand{\nonl}{\renewcommand{\nl}{\let\nl\oldnl}}%

\begin{document}






\vldbTitle{Distributed Edge Partitioning for Trillion-edge Graphs}
\vldbAuthors{Masatoshi Hanai, Toyotaro Suzumura, Wen Jun Tan, Elvis Liu, Georgios Theodoropoulos and Wentong Cai}
\vldbDOI{https://doi.org/10.14778/3358701.3358706}
\vldbVolume{12}
\vldbNumber{13}
\vldbYear{2019}

\title{Distributed Edge Partitioning for Trillion-edge Graphs}




\numberofauthors{3}
\author{
    \alignauthor Masatoshi Hanai \textsuperscript{\normalsize{1}}\thanks{This work is initiated when Dr. Hanai was in NTU.}\\
    \email{mhanai@acm.org}
    \alignauthor Toyotaro Suzumura \textsuperscript{\normalsize{2}}\\  
    \email{suzumura@acm.org}
    \alignauthor Wen Jun Tan \textsuperscript{\normalsize{3}}\\
    \email{wtan047@e.ntu.edu.sg}
    \and
    \alignauthor Elvis Liu \textsuperscript{\normalsize{1}}\\
    \email{esyliu@sustc.edu.cn}
    \alignauthor \mbox{Georgios Theodoropoulos \textsuperscript{\normalsize{1}}}\thanks{Corresponding author.}\\
    \email{georgios@sustc.edu.cn}
    \alignauthor Wentong Cai \textsuperscript{\normalsize{3}}\\
    \email{aswtcai@ntu.edu.sg}
    \end{tabular}
    \begin{tabular}{c}
    \affaddr{\textsuperscript{\normalsize{1}}Southern University of Science and Technology, Shenzhen, China} \\
    \affaddr{\textsuperscript{\normalsize{2}}IBM T.J. Watson Research Center, New York, USA} \\
    \affaddr{\textsuperscript{\normalsize{3}}Nanyang Technological University, Singapore} \\
}







\maketitle

\begin{abstract}
We propose Distributed Neighbor Expansion (Distributed NE), a parallel and distributed graph partitioning method that can scale to trillion-edge graphs while providing high partitioning quality.
Distributed NE is based on a new heuristic, called parallel expansion, where each partition is constructed in parallel by greedily expanding its edge set from a single vertex in such a way that the increase of the vertex cuts becomes local minimal.
We theoretically prove that the proposed method has the upper bound in the partitioning quality.
The empirical evaluation with various graphs shows that the proposed method produces higher-quality partitions than the state-of-the-art distributed graph partitioning algorithms.
The performance evaluation shows that the space efficiency of the proposed method is an order-of-magnitude better than the existing algorithms, keeping its time efficiency comparable.
As a result, Distributed NE can partition a trillion-edge graph using only 256 machines within 70 minutes.
\end{abstract}



%
%



\section{Introduction}\label{sec:intro}
Graph partitioning plays a critical role to efficiently analyze large-scale graphs on distributed graph-processing systems, such as Pregel, Giraph, GraphX, GPS, PowerGraph, Pregel+, PGX.D, X-Stream, ScaleGraph, GraM, PowerLyra, Gemini, and LA3~\cite{malewicz2010pregel,giraph,gonzalez2014graphx,low2012distributed,joseph2012powergraph,roy2013x,khayyat2013mizan,salihoglu2014optimizing,suzumura2015scalegraph,yan2015effective,hong2015pgx,wu2015g,Chen:2015:PDG:2741948.2741970,zhu2016gemini,Ahmad:2018:LSL:3204028.3228395}.
Graph partitioning aims to divide the input graph into parts in such a way that the communication cost among the distributed processes becomes minimal while keeping the distributed load balanced.

\begin{figure}
  \centering
    \subfigure{\includegraphics[width=.25\columnwidth]{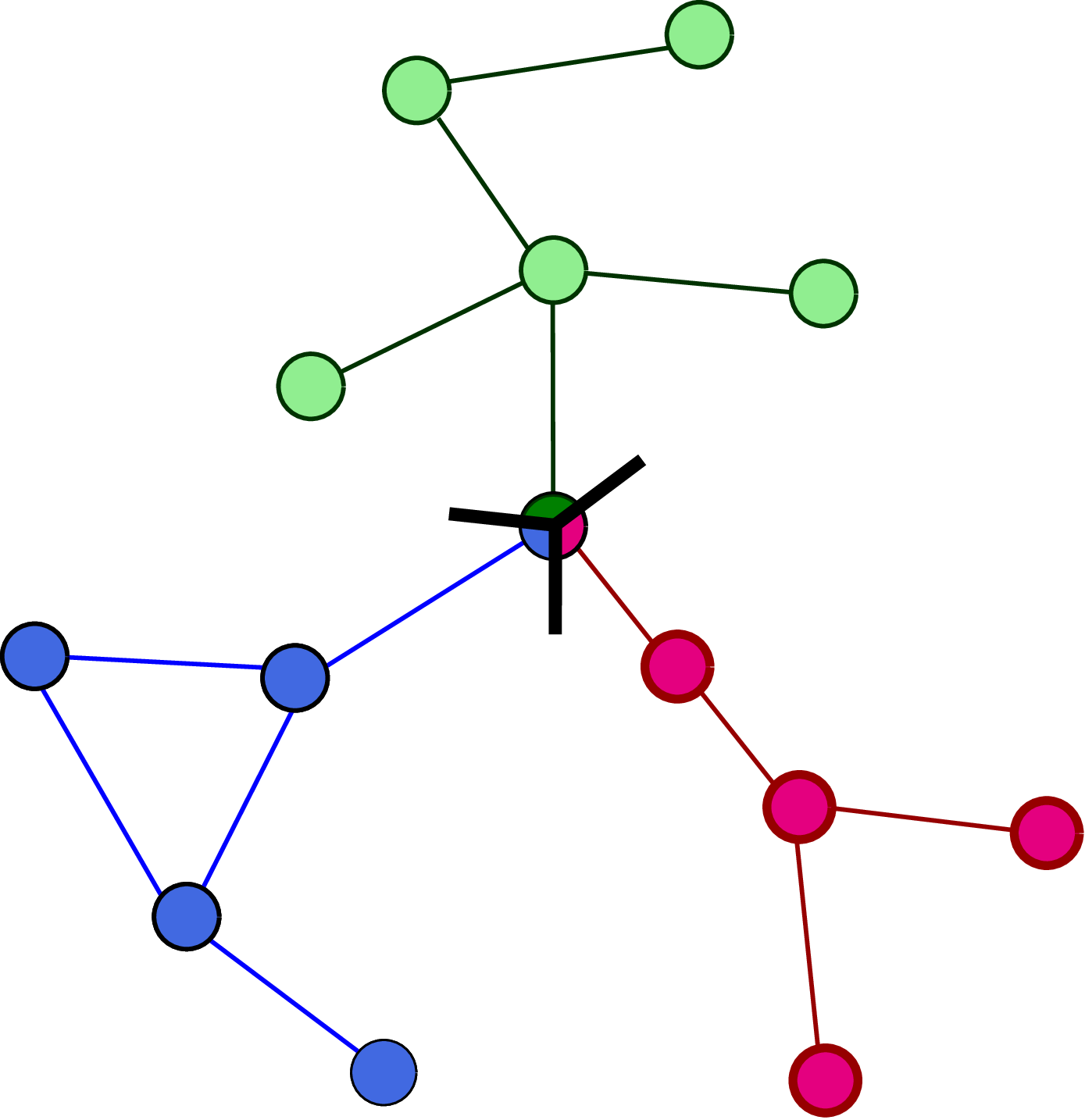}\label{fig:vertexcut}}
  \hspace{.05\textwidth}
    \subfigure{\includegraphics[width=.25\columnwidth]{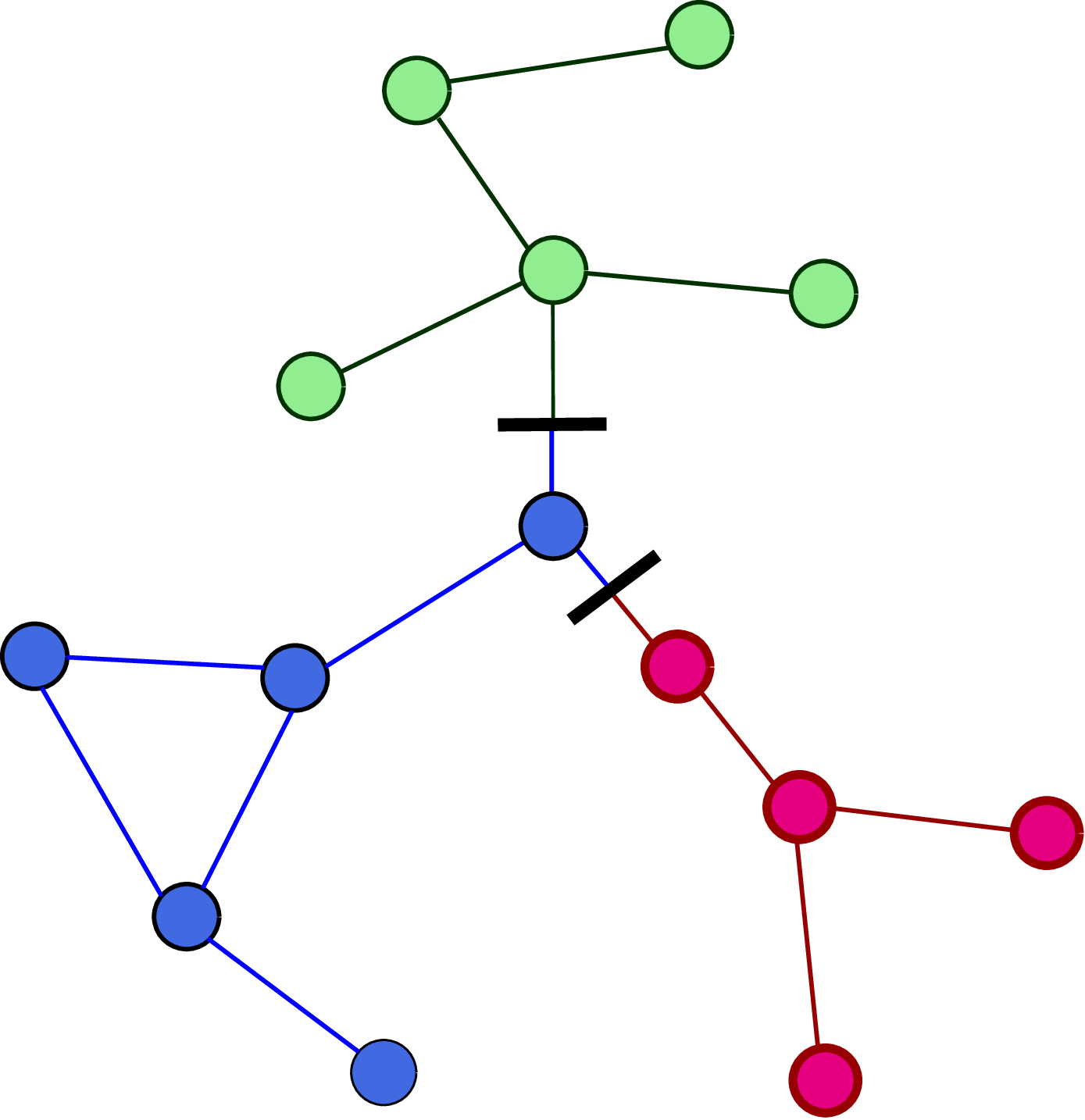}\label{fig:edgecut}}%
  \vspace{-10pt}
  \caption{(a) Edge Partition (Vertex-Cut Partition) vs (b) Vertex Partition (Edge-Cut Partition).}\label{fig:graph}
  \vspace{-18pt}
\end{figure}

In recent research~\cite{joseph2012powergraph,roy2013x,gonzalez2014graphx}, \emph{edge partitioning} is shown to be more effective than the traditional vertex partitioning in \begin{highlight}large-scale\end{highlight} real-world graphs, such as web or social graphs.
The edge partitioning divides the edges of the entire graph into disjoint parts as shown in Figure~\ref{fig:vertexcut}, whereas the vertex partitioning divides the vertices disjointly as shown in Figure~\ref{fig:edgecut}.
The edge partitioning provides the better workload balance because the computational cost of the graph processing essentially depends on the number of edges rather than the number of vertices.
The vertex partitioning causes serious workload imbalance between high-degree vertices and low-degree vertices in the \begin{highlight}large-scale\end{highlight} real-world graph, which mostly has skewed-degree distribution, namely, there are a few high-degree vertices, whereas the rest of the vertices have low degree
\begin{highlight}%
\cite{Leskovec2010} (e.g., web~\cite{Kleinberg1999} and social graphs~\cite{Mislove2007}). In this paper, we focus on the edge partitioning of such skewed graphs.
\end{highlight}

The research challenges of the edge partitioning fall into two different key issues: quality and scalability.
First, high \emph{partitioning quality}, which is measured by the total number of vertex cuts, is difficult to obtain since the minimization of the vertex cuts is proved to be an NP-hard problem~\cite{Zhang:2017:GEP:3097983.3098033}.
Second, the edge partitioning algorithm is required to scale to deal with trillion-edge graphs since the real-world graphs have been growing larger and larger, e.g., the social graph in Facebook consists of over one trillion edges~\cite{ching2015one}.
The targeting scale of the recent graph analysis has been sifting from the billion-edge to the trillion-edge scale~\cite{ching2015one,wu2015g,maass2017mosaic}.

However, no existing edge partitioning algorithms satisfy both the quality and scalability requirements for the trillion-edge graph.
On one hand, the scalable hash-based methods, such as 1D hash, 2D hash, DBH~\cite{NIPS2014_5396}, Hybrid Hash~\cite{Chen:2015:PDG:2741948.2741970}, produce the low-quality partitions because the edges are generally allocated at random by a hash operation.
Although some iterative local refinements of the random allocation are proposed to improve the quality, such as Oblivious~\cite{joseph2012powergraph} and Hybrid Ginger~\cite{Chen:2015:PDG:2741948.2741970}, their improvement is still limited.
On the other hand, the high-quality edge partitioning algorithms have a limit in their scalability.
Sequential algorithms~\cite{Petroni:2015:HSP:2806416.2806424,Zhang:2017:GEP:3097983.3098033} are not applicable for the trillion-edge graph because of the performance.
Even in the state-of-art high-quality distributed method, Sheep~\cite{Margo:2015:SDG:2824032.2824046} cannot scale to the trillion-edge graph because the speed up as the increase of machines is very limited.

\smallskip
In this paper, we propose \emph{Distributed Neighbor Expansion (Distributed NE)}, a novel distributed edge partitioning algorithm to solve the quality and scalability issues.
Distributed NE is based on a new simple greedy heuristic, called \emph{parallel expansion} as illustrated in Figure~\ref{fig:parallelexpansion}.
The main idea is that the algorithm starts from multiple random vertices and greedily expands each edge set in parallel such that the increase of the vertex cuts becomes minimal.
Such a greedy approach provides higher-quality edge partitions because it allocates most edges in a locally optimal way and seldom includes the random allocation.

\begin{figure}
  \begin{center}
  \includegraphics[width=.73\columnwidth]{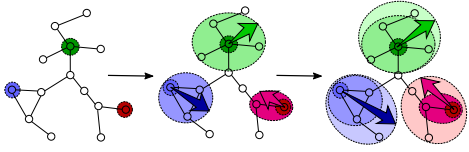}
  \end{center}
  \vspace{-13pt}
  \caption{Parallel Expansion. 3 parts greedily expand in parallel to partition a graph.}
  \label{fig:parallelexpansion}
  \vspace{-13pt}
\end{figure}

The key contributions of this paper are as follows:

\smallskip
\noindent\textit{\textbf{A Scalable and Efficient Algorithm.}}
Distributed NE is scalable and efficient.
We carefully design the distributed algorithm (\S~\ref{sec:proposal}) and propose a novel distributed edge allocation to solve the scalability and performance issues (\S~\ref{sec:implementation}).
We then propose a further optimization method which speeds up Distributed NE while keeping the quality (\S~\ref{sec:multiex}).
Due to the scalable and efficient design, our implementation can handle the trillion-edge graph using only a few hundreds of machines (\S~\ref{sec:evaluation}).

\smallskip
\noindent\textit{\textbf{High-Quality Partition with Theoretical Bound.}}
Distributed NE provides high-quality partitioning, which has a theoretical upper bound.
Our theoretical analysis proves that the number of the vertex cuts computed by Distributed NE is bounded from above.
Its upper bound is shown to be better than that of the existing distributed methods (\S~\ref{sec:theoretical}).
Moreover, we empirically show that Distributed NE produces significantly higher-quality partitions than the existing distributed methods in various skewed graphs (\S~\ref{sec:evaluation}).



\vspace{-6pt}
\section{Background}\label{sec:background}
This section provides a formulation of the problem and an outline of related work.
\vspace{-3pt}
\subsection{Notation and Problem Definition}
Let $G(V, E)$ be an undirected and unweighted graph which consists of a set of vertices, $V$, and edges, $E$.
The vertex set involved in $E$ is defined as $V(E)$.
An edge, $e$ $(\in E)$, connecting vertices, $v$ and $u$, is represented by $e_{v,u}$.
Let $p$ be a partition id and $P$ be the set of partition ids.
The number of elements in a set is represented by $|\cdot|$, e.g., $|V|$ and $|P|$. 

The objective of the edge partitioning is to divide $E$ into the disjoint subsets $E_p$ ($\bigcup_{p \in P} E_p = E$) such that the number of vertex cuts becomes as minimal as possible while keeping the balance of the subsets.
We usually measure the number of \emph{vertex replications}, which is defined as $\sum_{p \in P} |V(E_p)|$, instead of vertex cuts.
The number of vertex replications is normalized as follows:
\begin{equation}
    \frac{1}{|V|} \sum_{p \in P} |V(E_p)|,
    \label{eq:replicationfactor}
\end{equation}
which is called \emph{replication factor}~\cite{joseph2012powergraph}.
Based on the replication factor, the objective of \emph{a balanced $|P|$-way edge partitioning of $G$} \cite{joseph2012powergraph} is formalized as follows:
\begin{equation}
  \min_{f \in \mathcal{F}} \frac{1}{|V|} \sum_{E_p \in f(E)} |V(E_p)|,\ \ \ \text{s.t.}\ \max_{p \in P} |E_p| < \alpha \frac{|E|}{|P|}, \label{eq:problemdef}
\end{equation}
where $f: E \mapsto \{E_p: p \in P\}$ is a partitioning method, and $\mathcal{F}$ is the set of all partitioning methods.
The imbalance factor, $\alpha \geq 1.0$, is a constant parameter.

\subsection{Related Work}
The graph partitioning problem has been extensively investigated.
We refer the reader to the comprehensive surveys by Bulu\c{c} et al.~\cite{bulucc2016recent}.
Here, we provide a review of the edge partitioning methods and the vertex partitioning methods that have been proposed for trillion-edge graphs.

\smallskip
\noindent\textbf{Edge Partitioning.}
One of the major approaches to edge partitioning is based on the random hash.
The most straightforward approach is 1D-hash partitioning, where the edge is randomly assigned to a one-dimensional partitioning space.
Another approach uses 2D-hash partitioning, where the edge is assigned to two-dimensional partitioning space by hashing the adjacent vertices separately~\cite{Yoo:2011:SEL:2063384.2063469,Boman:2013:SMC:2503210.2503293,doi:10.1137/080737770,gonzalez2014graphx}.
The latest hash-based approaches utilize the degree of vertices, where the edge is randomly assigned so that high-degree vertices are divided into more partitions than low-degree ones.
The examples are Hybrid Hashing~\cite{Chen:2015:PDG:2741948.2741970} and DBH~\cite{NIPS2014_5396}.
Oblivious~\cite{joseph2012powergraph} and Hybrid Ginger~\cite{Chen:2015:PDG:2741948.2741970} are a heuristic to iteratively refine the assignment after the hash partitioning for the improvement of partitioning quality.
Since the hash operation is lightweight, this approach efficiently partitions the trillion-edge graph, but the random allocation causes the serious quality loss.

A more recent approach for the large-scale graph is based on streaming methods, such as FENNEL-based edge partitioning~\cite{Tsourakakis:2014:FSG:2556195.2556213,Bourse:2014:BGE:2623330.2623660}, HDRF~\cite{Petroni:2015:HSP:2806416.2806424}, and SNE~\cite{Zhang:2017:GEP:3097983.3098033}, where the input graph is represented as a sequence of edges and processed one-by-one without the entire information.
Since only a part of the entire graph is deployed on the main memory, it can handle a large graph whose size is more than the main memory. 
However, since the algorithm does not fully make use of the information of the entire graph, the partitioning quality is limited.
Moreover, these methods basically assume the sequential processing, and thus, their performance is limited.

Sheep~\cite{Margo:2015:SDG:2824032.2824046} is the state-of-the-art distributed edge partition method, where the graph is parallelly translated into the elimination tree before applying tree partitioning.
Although the algorithm can be executed on the distributed memory, it does not speed up well as the increase of the machines as reported in the paper.
Moreover, the algorithm does not provide any theoretical guarantee of the partitioning quality, such as the upper bound.

Most related work to ours is a greedy approach~\cite{Bourse:2014:BGE:2623330.2623660,Zhang:2017:GEP:3097983.3098033}.
NE~\cite{Zhang:2017:GEP:3097983.3098033} is the state-of-the-art greedy algorithm based on the expansion of the edge set. It currently provides the best quality in practice, but the scalability is limited since it is an offline sequential algorithm, where the entire graph is deployed on the main memory on a single machine.

\smallskip
\noindent\textbf{Vertex Partitioning of Trillion-edge Graphs.}
Recently, a few distributed vertex partitioning methods which can handle the trillion-edge graph have been proposed.
XtraPuLP~\cite{7967155} can scale to the trillion-edge graph, but it requires over 8K machines, which are excessive because PageRank of Facebook's trillion-edge graph can be calculated on only a few hundreds of machines~\cite{ching2015one}.
Spinner~\cite{7930049} is reported to partition the trillion-edge graph by using a few hundreds of machines, but it generates low-quality partitions because it involves the initial random allocation in the same way as the hash-based edge partitioning methods.
Moreover, these methods do not give the theoretical guarantee of the partitioning quality.

\vspace{-5pt}
\section{Distributed Neighbor Expansion} \label{sec:proposal}
In this section, we first illustrate the greedy expansion approach, which forms the basis of our distributed algorithm.
Second, we summarize challenging issues towards the scalable distributed algorithm.
Finally, we give an overview of our proposal.
\subsection{Edge Partitioning Based on Expansion} \label{subsec:expansion}
We first provide a short overview of the expansion of an edge set for partitioning.
After that, we review two greedy heuristics for effective expansion.

\begin{figure}
  \centering
  \includegraphics[width=.85\columnwidth]{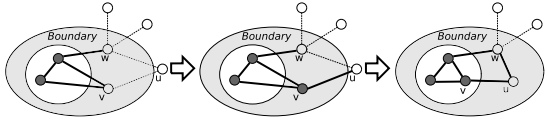}
  \vspace{-10pt}
  \caption{Expansion Approach. First heuristic selects $v$ from the boundary and allocates $e_{v,u}$. Then, second heuristic allocates $e_{u,w}$. This figure is modified from~\protect\cite{Zhang:2017:GEP:3097983.3098033}.}
  \label{fig:ne-example}
  \vspace{-10px}
\end{figure}

\begin{figure*}[t]
  \begin{center}
  \includegraphics[width=1.85\columnwidth]{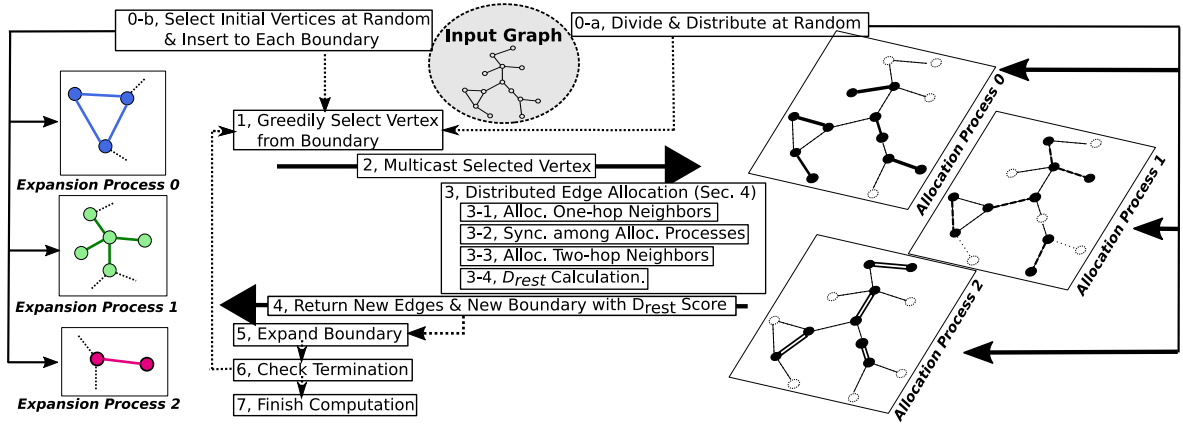}
  \vspace{-5pt}
  \caption{Overview of Distributed NE. A dotted and solid arrow represent a work and data flow, respectively.}
  \label{fig:ex-alloc-process}
  \end{center}
  \vspace{-18px}
\end{figure*}

\smallskip
\noindent\textbf{Edge Partitioning Using Expansion of Edge Set.}
For $G(V,E)$, let $X$ be an edge set ($X \subset E$).
We refer to the vertex set $B(X) := \{v\ |\ v \in V(X) \land \exists e_{v,u} \in E \backslash X\}$ as \emph{boundary} of $X$.
The \emph{expansion of $X$} is an operation defined as follows:
\begin{equation*}
\textit{Select}\ v \in B(X),\ \textit{and}\ X \leftarrow X \cup \{e_{v,u}\ |\ e_{v,u} \in E \setminus X\},
\end{equation*}
where we also suppose that if $B(X)$ is empty, $v$ is selected at random from $V(E \setminus X)$.

Based on the expansion, we can simply generate $|P|$ edge partitions of $G(V,E)$ as the following algorithm: 
\begin{itemize}
\vspace{-3pt}%
\item[(i)] The computation for each partition begins with a single boundary vertex and an empty edge set.
\vspace{-5pt}%
\item[(ii)] The edge set is expanded until its size reaches to the limit (i.e., $\alpha |E|/|P|$ as discussed in Section~\ref{sec:background}) or there is no edge for expansion. 
The generated edge set is one of the edge partitions, and it is removed from $E$.
\vspace{-5pt}%
\item[(iii)] The next partition is computed for the remaining edges $E$. If $E$ is empty, the algorithm terminates.
\vspace{-3pt}%
\end{itemize}
There are two basic heuristics~\cite{joseph2012powergraph,Bourse:2014:BGE:2623330.2623660,Petroni:2015:HSP:2806416.2806424,Chen:2015:PDG:2741948.2741970,Zhang:2017:GEP:3097983.3098033} to obtain high-quality edge partitions by using the expansion. 
One aims to select a vertex from the boundary in such a way that the increase of vertex replications becomes minimal.
The other is used during expansion to allocate additional edges which do not increase the replication vertices anymore.
Figure~\ref{fig:ne-example} shows the overview of how the edge set is expanded by using these heuristics.

\smallskip
\noindent\textbf{Heuristic for Vertex Selection from Boundary.}
In the above algorithm, let $E_p(t)$ be the set of $p$'s edge set allocated so far until $t$, where $t$ is a counter of the expansions.
Based on Formula~\eqref{eq:replicationfactor}, the number of $p$'s vertex replications at $t$ is represented as $|V(E_p(t))|$.



The increase of the vertex replications by the new boundary vertex $v$ from $t$ to $t+1$ is equal to the $v$'s degree in the remaining edges.
\begin{equation}\label{eq:drest}
    |V(E_p(t+1))| - |V(E_p(t))| = D_{\mathit{rest}}(v),
\end{equation}
where $D_{\mathit{rest}}(v) := |N(v) \cap E|$; $N(v)$ is a set of $v$'s neighbor edges; and $E$ is the remaining edges.
Thus, new vertex which minimizes the increase of the vertex replications, $v_{min}$, is represented as follows:
\begin{equation}\label{eq:vmin}
    v_{min} := \argmin_{x \in B} D_{\mathit{rest}}(x).
\end{equation}

For example in Figure~\ref{fig:ne-example}, \begin{highlight}$D_{\mathit{rest}}(w)$ is 3 \end{highlight}, and thus the number of vertex replications would increase by 3 if $w$ is selected; whereas $D_{\mathit{rest}}(v)$ is 1, and the number of replications increases by 1 if $v$ is selected.
Finally, $v$ is selected in this case (i.e., $v_{min} = v$).

\smallskip
\noindent\textbf{Heuristic for Edge Allocation.}
In the above algorithm, two-hop neighbors of the selected vertex may include some edges which never increase the replication vertices.
Such edges can be easily found with a simple yet important fact, that is commonly used in the existing algorithms~\cite{joseph2012powergraph,Bourse:2014:BGE:2623330.2623660,Petroni:2015:HSP:2806416.2806424,Chen:2015:PDG:2741948.2741970,Zhang:2017:GEP:3097983.3098033}, as follows:
\begin{flalign}\label{fact}
&\nonumber\textit{If both vertex $v$ and $u$ are involved in $V(E_p(t))$, then}\\
&\textit{allocation of $e_{v,u}$ to $p$ does not increase replications.}
\end{flalign}


\noindent By this fact, when allocating an edge to $p$, if a two-hop-neighbor vertex from the selected vertex is already allocated to $p$, then the edge from the one-hop-neighbor vertex to the two-hop-neighbor vertex is allocated.

For example in Figure~\ref{fig:ne-example}, $e_{u,w}$ is additionally allocated during the expansion.
The allocation of $e_{u,w}$ does not increase the number of the vertex replications because one of the edges to $w$ and one of the edges to $v$ are already allocated to the same partition.

\subsection{Challenges to Distributed Algorithm}
The main idea behind our proposed approach is to execute the expansion for each partition in parallel.
Even though the core idea of the parallel expansion is simple, to achieve a highly scalable distributed algorithm which enables us to partition trillion-edge graphs is a challenging endeavour.
The challenges are summarized as follows:

\smallskip
\noindent\textbf{Scalability.}
The primal issue is constraints on the scalability of the partitioning algorithm imposed by the limited memory size.
The parallel expansion is an offline algorithm, where the entire input graph is necessary to be accessible from the beginning of the algorithm.
The effective initial deployment of the input graph to the distributed main memories is required for scalability and efficiency.




\smallskip
\noindent\textbf{Concurrent Allocation.}
Concurrency issues appear during the parallel expansion when multiple expanded parts simultaneously try to allocate the same graph element (i.e., vertex or edge).
If the element is replicated among the distributed processes, the global synchronization of the allocated information is necessary among the replications to maintain consistency.

\smallskip
\noindent\textbf{Distributed Two-hop-neighbor Search.}
Another issue is introduced by a two-hop-neighbor search in the greedy expansion.
The search of the two-hop neighbors from the current expanded part is necessary to greedily decide which edges should be allocated next.
However, in the general distributed graph, each vertex has basically only its one-hop-neighbor vertices and edges (i.e., adjacent vertices and edges), and the two-hop-neighbor search cannot be fully localized.
It must include the communication among the distributed processes.
The efficient two-hop-neighbor search is a key challenge, which differentiates Distributed NE from the other typical distributed graph algorithms based on the one-hop-neighbor accesses.

\subsection{Proposed Distributed Approach}
In this subsection, Distributed NE is presented.
Suppose there are $|P|$ machines available.
\begin{highlight}Since the graph partitioning is used for the pre-processing of the distributed graph applications (such as shortest path and PageRank), the number of partitions generated by Distributed NE should be equal to the most effective number for the distributed graph processing systems, which we assume are the most efficient when each partition is assigned to one machine.\end{highlight}
Thus, the goal of the task is specified to be:

\smallskip\noindent
\textit{To compute $|P|$-way edge partitioning of the input graph and distribute the $|P|$ partitions into the $|P|$ machines.}
\smallskip


For the computation of the parallel expansion on the distributed environment, there are two basic requirements. 
First, the entire input graph needs to be initially stored in the main memories on the multiple machines for the graph traversal during the expansion.
Second, the boundary vertices of the edge set in each partition need to be managed in a single machine because the selection of the local-minimal vertex from the boundary is executed locally for efficiency.

To meet the two requirements, two types of a distributed process are used: an \emph{expansion process} and an \emph{allocation process}.
The expansion processes manage the boundary vertices to compute local-optimal vertex selection.
The allocation processes manage the input graph in the distributed way and manage the allocation of vertices and edges.

The flow of the computation is illustrated in Figure~\ref{fig:ex-alloc-process}.
The overall computation begins with launching $|P|$ expansion processes and $|P|$ allocation processes (i.e., Expansion Process 0,1,2 and Allocation Process 0,1,2 in Figure~\ref{fig:ex-alloc-process}).
Each expansion process and allocation process is deployed to one of the $|P|$ machines.
First, the input graph is randomly distributed to the allocation processes.
The expansion in each expansion process begins with a random-selected vertex.
The algorithm is executed iteratively as illustrated by Step 1--6 in Figure~\ref{fig:ex-alloc-process}.
During the iterations, the allocated edges are copied and sent from the allocation processes to the expansion processes, and at the end of the computation, the entire edges are distributed to the $|P|$ expansion processes.

\removelatexerror
\setlength{\textfloatsep}{0pt}
\begin{algorithm}
\caption{Expansion Process for Partition $p$}\label{alg:parallelpartition}
\SetKwInOut{Var}{Var}
\SetKwInOut{VarEmpty}{}
\SetKwInput{Input}{Input}
\SetKwInput{Output}{Output}
\Input{$p$ -- This Partition ID, $P$ -- Set of Partition IDs}
\Output{$E_{p}$ -- Partitioned Edges for $p$}
\BlankLine
\Var{$B_p$ -- Priority Queue of $\langle D_{rest}(v), v \rangle$}
\VarEmpty{$|E|$ -- Total \# of Allocated Edges}
\VarEmpty{$|E_{\mathit{init}}|$ -- Total \# of Initial Edges}
\SetKw{Break}{break}

\DontPrintSemicolon
\BlankLine

$E_{p}$ $\leftarrow$ $\varnothing$, $B_{p}$ $\leftarrow$ $\varnothing$\;

\BlankLine

\While{$true$} {\label{line:startloop}
    \nonl /* Selection of Expansion Vertex */ \;
    $v_{\mathit{min}}$ $\leftarrow$ $\varnothing$ \;\label{line:startselect}
    \uIf {$B_{p} \neq \varnothing$} {
        $v_{\mathit{min}}$ $\leftarrow$ $B_{p}$.\texttt{popMinDrestVertex()} \;
    }
    \Else {
        $v_{\mathit{min}}$ $\leftarrow$ \texttt{getRandomVertex()} \;
    }\label{line:endselect}
    
    \BlankLine
    
    \nonl /* Request Allocation */ \;
    \texttt{MulticastVertexToAllocators($v_{\mathit{min}}$,$p$)} \; \label{line:startreq}
    \texttt{Barrier()} /* Wait for Allocation */\;
    
    \BlankLine
    \nonl /* Update Boundary and Edge Sets */ \;
    $B_{\mathit{new}}$ $\leftarrow$ \texttt{ReceiveNewBoundaryWithDrest()} \; \label{line:startupdate}
    $B_{p}$.\texttt{insert($B_{\mathit{new}}$)} \;
    $E_{\mathit{new}}$ $\leftarrow$ \texttt{ReceiveNewAllocatedEdges()} \; \label{line:recnewedge}
    $E_{p}$ $\leftarrow$ $E_{p}$ $\cup$ $E_{\mathit{new}}$ \; \label{line:endupdate} 
    
    \BlankLine
    \nonl /* Check Termination */ \;
    $|E|$ $\leftarrow$ \texttt{AllGatherSum($|E_{p}|$)} \;
    \lIf{$E_{p} > \alpha \cdot |E_{\mathit{init}}|/|P|$ \  $\lor$ \  $|E| = |E_{\mathit{init}}|$} {
        \Break
    }\label{line:endcheck}
    \BlankLine
}\label{line:endloop}
\end{algorithm}

\smallskip
\noindent\textbf{Expansion Process.}
Each expansion process computes in parallel the selection of a new vertex from the boundary, $B_p$ to expand the edge set, $E_p$.

Algorithm~\ref{alg:parallelpartition} shows the details of the computation in each expansion process.
%
Each partition is computed iteratively, including the vertex selection (Line~\ref{line:startselect} -- \ref{line:endselect}), the allocation request (Line~\ref{line:startreq}), and the update of the boundary and the edge set (Line~\ref{line:startupdate} -- \ref{line:endupdate}).
In the vertex selection, a vertex, $v_{\mathit{min}}$, which has the minimal degree for the remaining edges as discussed in Equation~\eqref{eq:vmin}, is selected from $B_p$ if it is not empty; otherwise, $v_{\mathit{min}}$ is chosen at random from vertices with non-allocated edges.
The random vertex is basically taken from the allocation process in the same machine. It is from the other machines only if necessary.
After that, in the allocation request, each expansion process multicasts $v_{\mathit{min}}$ to allocation processes in charge, and waits for the edge allocation by the allocation processes.
When the edge allocation is done, each expansion process receives new boundary, $B_{\mathit{new}}$, which also includes $D_{\mathit{rest}}$ score of each vertex as discussed in Equation~\eqref{eq:drest}, and newly allocated edges, $E_{\mathit{new}}$.
$B_{\mathit{new}}$ and $E_{\mathit{new}}$ are inserted to $B_p$ and $E_p$ respectively. 
Finally, the expansion process checks for the termination condition and stops accordingly.

\smallskip
\noindent\textbf{Allocation Process.}
The allocation processes take the selected vertices from the expansion processes and return newly allocated edges and new boundary vertices with their $D_{\mathit{rest}}$ scores.
For each selected vertex $v$ of partition $p$, two types of edges are allocated as discussed in Section~\ref{subsec:expansion}.
\vspace{-3pt}%
\begin{itemize}
\item[(i)] $v$'s one-hop neighbor edges, $\{e_{v,u}\}$.
\vspace{-5pt}%
\item[(ii)] $v$'s two-hop neighbor edges which do not increase the vertex replications, namely, which satisfy Condition~\eqref{fact}.
\end{itemize}%
\vspace{-3pt}%

Unlike the single machine environment, where these edges are straightforwardly searched and allocated, the allocation on the distributed environment causes two issues: the management of concurrent allocation and the two-hop neighbor search.
In the next section, we propose a new technique to solve the issues.

\section{Distributed Edge Allocation} \label{sec:implementation}
In this section, an edge allocation algorithm to utilize in the distributed allocation processes is proposed.
We first present a basic data structure which is used to divide the input graph and to store the divided graphs.
The input graph needs to be evenly divided for workload balancing and to be stored in a space-efficient way.
Second, we present an algorithm which efficiently allocates one-hop and two-hop neighbors while solving concurrency problems.

\smallskip
\noindent\textbf{Data Structure.}
In the allocation processes, the input graph is distributed by the edge partitioning algorithm, which uniquely divides the edges and replicates the vertices among the allocation processes.
Specifically, we use 2D-hash partitioning and compressed sparse row (CSR)~\cite{Yoo:2011:SEL:2063384.2063469,Boman:2013:SMC:2503210.2503293,doi:10.1137/080737770,gonzalez2014graphx} for the initial distribution of the input graph.
In 2D-hash partitioning, the metadata of replicated vertices (i.e., process ids of each vertex) can be calculated from vertex id instead of storing pairs of vertex and process id, which suppresses memory space in the case of trillion-edges graphs.

On the edge-partitioned graph, global synchronization is executed only in the vertex replications, whereas the edge's allocation information does not need to synchronize since it is unique.
For example in Figure~\ref{fig:ex-alloc-process}, the input edges are uniquely divided into three processes (solid line, dotted line and double line).
The vertices are replicated to the three processes and synchronize their allocation information during the computation.

There are two main advantages of this vertex-replicated approach compared to other approaches, such as replicating edges only or replicating both edges and vertices hybridly.

First, this approach can locally solve the allocation conflict, which occurs only in the edge and not in the vertex.
On one hand, the edge allocation may result in conflicts because each edge is allocated by the unique partition.
Thus, if the multiple partitions concurrently try to allocate the same edge, then conflict resolution is necessary to decide which partition actually allocates the edge.
On the other hand, the vertex allocation does not cause any conflict because a vertex may be allocated by the multiple partitions.
If the multiple partitions concurrently allocate the same vertex, these partitions can be simply assigned to the vertex.

Second, the replication of vertices is space efficient compared to the replication of edges because skewed graphs have much more edges than vertices (e.g., Facebook graph~\cite{ching2015one} contains 1.45B vertices and 1T edges). 
Replications of vertices generally require less memory than that of edges.

\smallskip
\noindent\textbf{Algorithm.}
The key idea of the algorithm is to first allocate the one-hop-neighbor edges and then allocate two-hop-neighbor edges which satisfy Condition~\eqref{fact}.
Between the one-hop-neighbor and two-hop-neighbor allocation, synchronization is executed to detect two-hop neighbors on the other processes and to make the allocation information consistent among the replications.
The conflicts of the edge allocation are solved locally.
Algorithms~\ref{alg:parallelallocation} and \ref{alg:details} provide a high level and a detailed description of the algorithm respectively.
Figure~\ref{fig:distributededgeallocation} illustrates the edge allocation for the graph in Figure~\ref{fig:graph}.

\removelatexerror
\begin{algorithm}[t]
\DontPrintSemicolon
\caption{Overview of Edge Allocation}\label{alg:parallelallocation}
\SetKw{Break}{break}
\SetKw{Continue}{continue}
\SetKwProg{Fn}{}{}{}
\SetKwInOut{Args}{Args}
\Args{
    \begin{tabularx}{\textwidth}[t]{X}
        $\mathit{VP}_{\mathit{rec}}$ -- Received $\{ \langle v_{\mathit{min}}, p \rangle \}$ via Line~\ref{line:startreq} in Alg.~\ref{alg:parallelpartition}\\
    \end{tabularx}
}
\BlankLine
\SetKwInOut{LocalV}{Var}
\LocalV{
    \begin{tabularx}{\textwidth}[t]{X}
        $\mathit{BP}_{\mathit{new}}$ -- New Boundary-Partition Pairs $\{ \langle v, p \rangle \}$\\
        $\mathit{EP}_{\mathit{1hop}}$ -- Alloc. 1Hop-Partition Pairs $\{ \langle e, p \rangle \}$ \\
        $\mathit{EP}_{\mathit{2hop}}$ -- Alloc. 2Hop-Partition Pairs $\{ \langle e, p \rangle \}$ \\
        $\mathit{LD_{rest}}$ -- Local $D_{rest}$ Scores $\{ \langle v, \mathit{D_{rest}(v)} \rangle \}$\\
    \end{tabularx}
}

\SetKwFor{ParFor}{for}{do in parallel}{end} 

\BlankLine
\nonl  /* This function is called after all the selected vertices \;
\nonl \ * are received. (After Line~\ref{line:startreq} in Algorithm~\ref{alg:parallelpartition}) */\;
\SetKwFunction{EdgeAlloc}{EdgeAllocation}
\Fn{\EdgeAlloc{$\mathit{VP}_{\mathit{rec}}$}}{
    $\langle \mathit{BP}_{\mathit{new}}^{\mathit{local}}$, $EP_{\mathit{1hop}} \rangle$ $\leftarrow$ \texttt{AllocteOneHopNeighbors($\mathit{VP}_{\mathit{rec}}$)} \; \label{line:eb}
    \BlankLine
    $\mathit{BP}_{\mathit{new}}$ $\leftarrow$ \texttt{SyncVertexAllocations($\mathit{BP}_{\mathit{new}}^{\mathit{local}}$)} \; \label{line:syncvertex}
    \BlankLine
    $\mathit{EP}_{\mathit{2hop}}$ $\leftarrow$ \texttt{AllocateTwoHopNeighbors($\mathit{BP}_{\mathit{new}}$)} \; \label{line:twohop}
    \BlankLine
    $\mathit{LD}_{\mathit{rest}}$ $\leftarrow$ \texttt{ComputeLocalDrest($\mathit{BP}_{\mathit{new}}$)} \; \label{line:score}
    \BlankLine
    \texttt{SendNewBoundaryWithLocalDrest($\mathit{BP}_{\mathit{new}}$, $\mathit{LD}_{\mathit{rest}}$)}\; \label{line:sndback_boundary}
    \texttt{SendNewAllocatedEdges($\mathit{EP}_{\mathit{1hop}}$ $\cup$ $\mathit{EP}_{\mathit{2hop}}$)} \; \label{line:sndback_newedges}
}
\end{algorithm}
\removelatexerror
\begin{algorithm}[!t]
\DontPrintSemicolon
\caption{Details of Algorithm~\ref{alg:parallelallocation}}\label{alg:details}
\SetKw{Break}{break}
\SetKw{Continue}{continue}
\SetKw{Return}{return}
\SetKwProg{Fn}{}{}{}
\SetKwInOut{Args}{Argument}
\SetKwFor{ParFor}{for}{do in parallel}{end} 

\SetKwInOut{SubGraph}{SubGraph}
\SubGraph{
    \begin{tabularx}{\textwidth}[t]{X}
        $\mathit{SubG}$ -- Initial Distributed Subgraph \\
    \end{tabularx}
}

\BlankLine
\SetKwFunction{AllocteOneHopNeighbors}{AllocteOneHopNeighbors}
\Fn{\AllocteOneHopNeighbors{$\mathit{VP}_{\mathit{rec}}$}}{
    $\mathit{BP}^{\mathit{local}}_{\mathit{new}}$ $\leftarrow$ $\varnothing$; $\mathit{EP}_{\mathit{1hop}}$ $\leftarrow$ $\varnothing$\;  
    \ParFor{$\langle v, p \rangle$ $\in$ $\mathit{VP}_{\mathit{rec}}$} {
        \For{$e_{v,u}$ $\in$ $\mathit{SubG}$.\texttt{NonAllocNeighbors($v$)}} {
            $\mathit{SubG}$.\texttt{AppendVertexPartition($u$, $p$)} \;
            $\mathit{SubG}$.\texttt{AllocateEdge($e_{v,u}$, $p$)} \;
            $\mathit{BP}^{\mathit{local}}_{\mathit{new}}$ $\gets$ $\mathit{BP}^{\mathit{local}}_{\mathit{new}}$ $\cup$ $\{\langle u,p \rangle\}$\;
            $\mathit{EP}_{\mathit{1hop}}$ $\gets$ $\mathit{EP}_{\mathit{1hop}}$ $\cup$ $\{\langle e_{v,u}, p \rangle\}$\;
        }
    }
    \Return $\langle \mathit{BP^{local}_{new}}$, $\mathit{EP}_{\mathit{1hop}} \rangle$
}

    

\BlankLine
\SetKwFunction{AllocTwoHop}{AllocateTwoHopNeighbors}
\Fn{\AllocTwoHop{$\mathit{BP}_{\mathit{new}}$}}{
    $\mathit{EP_{2hop}}$ $\leftarrow$ $\varnothing$ \;
    \ParFor {$\langle u, .\rangle$ $\in$ $\mathit{BP}_{\mathit{new}}$} {\label{line:twohop_begin}
        \For{$e_{u,w}$ $\in$ $\mathit{SubG}$.\texttt{NonAllocNeighbors($u$)}} {
            $P_{\mathit{new}}$ $\leftarrow$ $\mathit{SubG}$.\texttt{Parti($u$)} $\cap$ $\mathit{SubG}$.\texttt{Parti($w$)}\;
            \If{$P_{\mathit{new}}$ $\not=$ $\varnothing$} {
                $p_{\mathit{new}}$ $\leftarrow$ $\displaystyle \argmin_{x \in P_{\mathit{new}}} $ $SubG$.\texttt{NumEdges($x$)}\;
            
                \BlankLine
                $\mathit{SubG}$.\texttt{AllocateEdge($e_{u,w}$, $p_{\mathit{new}}$)} \;
                $\mathit{EP_{2hop}}$ $\gets$ $\mathit{EP_{2hop}}$ $\cup$ $\{ \langle e_{u,w}, p_{\mathit{new}} \rangle \}$\;
                
            }
        }
    }\label{line:twohop_end}
    \Return $\mathit{EP_{2hop}}$ \;
}
\end{algorithm}

\begin{figure*}[h]
  \begin{center}
  \includegraphics[width=1.8\columnwidth]{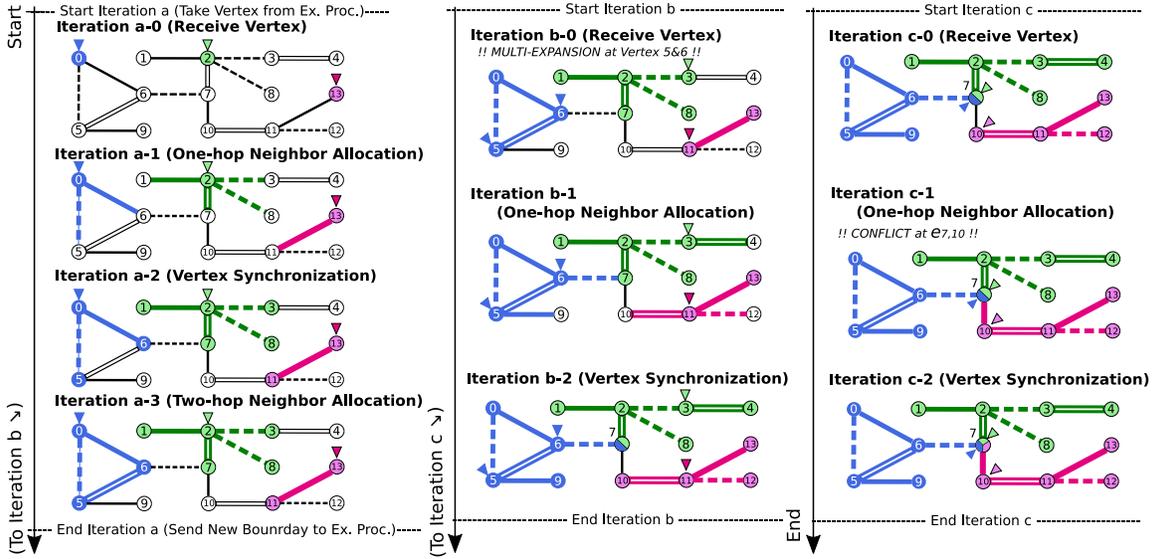}
  \vspace{-10pt}
  \caption{Distributed Edge Allocation. Each line type (solid, dotted, double) represents the initial partition for allocation process. Each color (blue, red, green) is the final partitioning result.}\label{fig:distributededgeallocation}
  \end{center}
  \vspace{-15px}
\end{figure*}

Specifically, \texttt{EdgeAllocation()} in Algorithm~\ref{alg:parallelallocation} is called after receiving all the selected vertices with their partition id, $\mathit{VP}_{\mathit{rec}}$, from the expansion processes (i.e., after all the expansion processes call \texttt{MulticastVertexToAllocators()} at Line~\ref{line:startreq} in Algorithm~\ref{alg:parallelpartition}).

The function consists of 4 phases as follows.

\smallskip
\noindent\textbf{(1) Expansion and One-hop-neighbor Allocation.} \\
In \texttt{AllocteOneHopNeighbors()}, for each received vertex and its associated partition $\langle v, p \rangle$, $v$'s one-hop-neighbor edge and vertex, $e_{v,u}$ and $u$, are allocated to $p$ in parallel as shown in Algorithm~\ref{alg:details}.
$\langle e_{v,u}, p \rangle$ and $\langle u, p \rangle$ are inserted to the new one-hop-neighbor edges, $\mathit{EP}_{\mathit{1hop}}$, and the new boundary, $\mathit{BP}^{\mathit{local}}_{\mathit{new}}$, respectively.
The conflict among different threads to allocate the same edge is solved by a CAS operation for performance.

For example in Figure~\ref{fig:distributededgeallocation}, at Iteration a-0, $v_0$ is received in the solid-line allocation process and the dotted-line one. Then, $e_{0,6}$ is allocated in the solid-line process, while $e_{0,5}$ is allocated in the dotted-line process. $v_6$ and $v_5$ become new boundaries for the next iteration. 
\begin{highlight}
In Iteration c-1, a conflict in the edge allocation occurs at $e_{7,10}$ if multiple threads simultaneously try to allocate $e_{7,10}$ to different partitions.
In this case, the edge is finally allocated to the red partition.
\end{highlight}

\smallskip
\noindent\textbf{(2) Synchronization of Vertex's Allocation Ids.} \\
\texttt{SyncVertexAllocations()} synchronizes vertices which are newly allocated in the previous one-hop-neighbor allocation, so that all vertex replications among the different machines have the same allocation ids. 
Specifically, the local data of the new boundaries, $\mathit{BP}^{\mathit{local}}_{\mathit{new}}$, are sent to their replications on the other machines.
Then, in the other machines, the allocation ids are appended to the associated vertices.
The synchronized allocation ids, $\mathit{BP}_{\mathit{new}}$, are used for the two-hop-neighbor allocation.
For example in Figure~\ref{fig:distributededgeallocation}, at Iteration a-2, $v_6$ in the solid-line process is synchronized with its replications in the dotted-line and the double-line ones. 

\smallskip
\noindent\textbf{(3) Two-hop-neighbor Allocation.} \\
In \texttt{AllocateTwoHopNeighbors()}, the neighbors of $\mathit{BP}_{\mathit{new}}$ (thus, two-hop neighbors of the received vertex, $v_{min}$) are processed in parallel as shown in Algorithm~\ref{alg:details}.
For a two-hop-neighbor edge, $e_{u,w}$, if both $u$ and $w$ have been already allocated to the same partitions, $P_{\mathit{new}}$, then the algorithm selects the minimal edge partition, $p_{\mathit{new}} \in P_{\mathit{new}}$, and allocates $e_{u,w}$ to $p_{\mathit{new}}$.
The new allocated edges are stored to $\mathit{EP_{2hop}}$ and sent back to the expansion processes.
For example in Figure~\ref{fig:distributededgeallocation}, $e_{5,6}$ at Iteration a-3 is allocated in this phase because both $v_5$ and $v_6$ are allocated to \textit{blue}.

\smallskip
\noindent\textbf{(4) Calculation of \begin{highlight}Local\end{highlight} $D_{\mathit{rest}}$ of New Boundaries.} \\
Finally, in \texttt{ComputeLocalDrest()}, the local score of $D_{\mathit{rest}}$ in each processes (as discussed in Equation~\eqref{eq:drest}) is calculated for each vertex in $\mathit{BP}_{\mathit{new}}$.
\begin{highlight}
At a later stage, the calculated local scores will be communicated to and gathered in the corresponding expansion process in order to compute the global $D_{\mathit{rest}}$ score.
\end{highlight}

\smallskip
After completing the 4 phases, the new boundary vertices, $\mathit{BP_{new}}$, with their local $D_{\mathit{rest}}$ score, $\mathit{LD_{rest}}$, and the new allocated edges, $\mathit{EP_{1hop}} \cup \mathit{EP_{2hop}}$, are sent back to the expansion processes.
In the expansion processes, they are received via \texttt{ReceiveNewBoundaryWithDrest()} at Line~\ref{line:startupdate} in Algorithm~\ref{alg:parallelpartition} and \texttt{ReceiveNewAllocatedEdges()} at Line~\ref{line:recnewedge}, respectively.
\begin{highlight}
Before calling \texttt{ReceiveNewBoundaryWithDrest()}, the local $D_{\mathit{rest}}$ scores for each vertex are summed up to calculate the global score in each expansion process.
\end{highlight}

\section{Iteration Reduction by Multi-Expansion}\label{sec:multiex}
In this section, we propose further performance optimization of Distributed NE by reducing the number of iterations in the expansion.
Even though our algorithm uses the efficient allocation method, performance bottleneck still appears in communications among the allocation processes because a large number of iterations are executed.
As a result, a large number of barrier operations are conducted, and it causes a performance problem.

Suppose each process evenly allocates $n$ edges per iteration,
the number of total iterations for computing the partition of $G = (V,E)$ is represented as $|E|\ /\ (n \cdot |P|)$.
Usually, the number of edges $|E|$ (e.g., 100K -- 1T) is much larger than the number of $|P|$ (e.g., up to 1K).
Therefore, the number of allocating edges per iteration $n$ is required to be large enough to reduce the number of iterations.

\vspace{-8px}
\begin{algorithm}[h]
\caption{Multi-Expansion (Corresponding to Line~\ref{line:startselect}--\ref{line:startreq} in Algorithm~\ref{alg:parallelpartition})}\label{alg:multiexpansion}
\SetKw{Break}{break}
\SetKwInOut{Variable}{Variable}
\Variable{$V_{\mathit{k\text{-}min}}$ -- Expanding Vertices}
\SetKwInOut{Parameter}{Parameter}
\Parameter{$\lambda$ -- Expansion Factor.}
\DontPrintSemicolon
\BlankLine
\setcounter{AlgoLine}{2}
\nonl \textbf{/* Multi-Expansion */} \;
$V_{\mathit{k\text{-}min}}$ $\leftarrow$ $\varnothing$ \label{line:mx:start}\;
\uIf {$B_{p} \neq \varnothing$} {
    \textbf{$k$ $\leftarrow$ $\lambda \cdot |B_p|$} \;
    \textbf{$V_{\mathit{k\text{-}min}}$ $\leftarrow$ $B_{p}$.\texttt{popK-MinDrestVertices($k$)}}\;
}
\Else {
    $V_{\mathit{k\text{-}min}}$ $\leftarrow$ Random Vertex \;
}
\nonl /* Request Allocation */ \;
\textbf{\texttt{MulticastVerticesToAllocators($V_{\mathit{k\text{-}min}}$,$p$)}} \label{line:mx:end}\;
\end{algorithm}
\vspace{-8px}

Algorithm~\ref{alg:multiexpansion} is the modified algorithm from Algorithm~\ref{alg:parallelpartition} in Section~\ref{sec:proposal}.
Line~\ref{line:startselect}--\ref{line:startreq} in Algorithm~\ref{alg:parallelpartition} are replaced into Line~\ref{line:mx:start}--\ref{line:mx:end} in Algorithm~\ref{alg:multiexpansion}.
To increase the number of allocating edges per iteration, Algorithm~\ref{alg:multiexpansion} expands multiple boundary vertices in one iteration.
Instead of selecting a single minimum scored vertex from the boundary, the algorithm selects $k$-minimum vertices per iteration.
The number of expanding vertices $k$ is determined by an \emph{expansion factor} $\lambda$ $(0 < \lambda \leq 1)$.
In the expansion, $k = \lambda \cdot |B_p|$ vertices are expanded for allocating new edges.
As $\lambda$ approaches $1$, the number of expanded edges per iteration increases, and the overall execution time is reduced.
On the other hand, the oversize $\lambda$ causes quality loss since the greedy heuristic may not work well.

\begin{figure}[h]
  \vspace{-5px}
  \centering
  \subfigure{\includegraphics[trim=2.2cm 0 2.5cm 0,clip,width=.4\columnwidth]{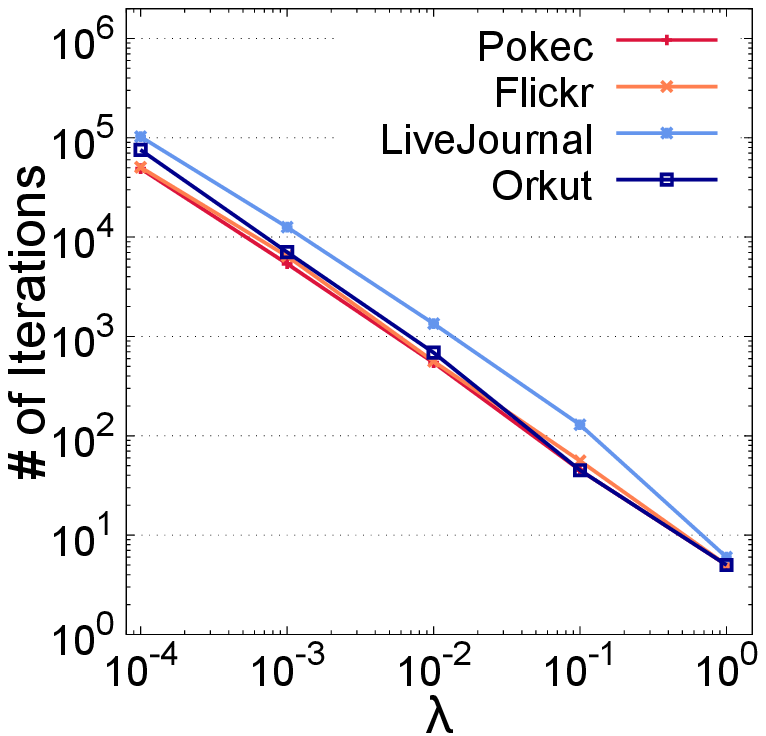}\label{fig:itr-lambda}} 
  \subfigure{\includegraphics[trim=2.2cm 0 2.5cm 0,clip,width=.4\columnwidth]{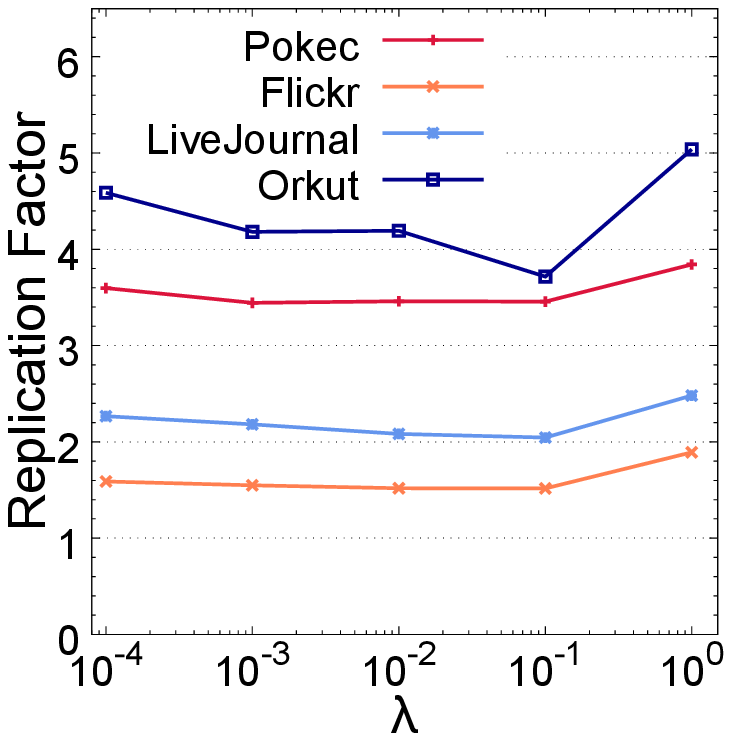}\label{fig:rf-lambda}}
  \vspace{-10px}
  \caption{Trend of \# of Iterations and Replication Factor on Different $\lambda$ (32 Partitions).}
  \label{fig:lambda}
  \vspace{-5px}
\end{figure}

Preliminary experiments have been conducted to select the value for $\lambda$.
Figure~\ref{fig:lambda} shows the relationship between the number of iterations and the replication factor (as discussed in Equation~\eqref{eq:replicationfactor}) under different $\lambda$ on 32 partitions.
The number of iterations linearly decreases as the increase of $\lambda$.
In $\lambda$ = $1.0$, which means all boundary vertices are chosen every iteration, the total number of iterations is less than $10$ in all graph data.
On the other hand, the replication factor slightly decreases from $\lambda$ = $10^{-4}$ to $10^{-1}$.
Then, in $\lambda$ = $1.0$, it becomes worse in all graph data.
We find the similar characteristics in the other numbers of partitions such as 4, 8, 16, and 64 partitions. 
Based on these results, we choose $\lambda$ = $0.1$ in our implementation to maximize the performance and quality.
\vspace{-5px}

\vspace{-5px}
\section{Theoretical Analysis}\label{sec:theoretical}
In this section, we first analyze a theoretical upper bound of the partitioning quality provided by Distributed NE (Theorem~\ref{thr:upperbound},~\ref{thr:tightness}).
\begin{highlight}We then provide the efficiency analysis of Distributed NE (Theorem~\ref{thr:time}).\end{highlight}

The partitioning quality given by Distributed NE (without the muti-expansion in Section~\ref{sec:multiex}) has the upper bound and satisfies the following two theorems.
The upper bound of partitioning quality by the sequential expansion has been obtained in~\cite{Zhang:2017:GEP:3097983.3098033}, but it does not satisfy the parallel case because the bound strongly assumes the sequential processing.
We provide new upper bound for the parallel case.
\begin{theorem}[Upper Bound]\label{thr:upperbound}
Edge partitions of a graph $G(V,E)$ computed by Distributed NE have replication factor ($RF$), such that
\vspace{-8px}
\begin{equation*}%
RF \leq \frac{|E| + |V| + |P|}{|V|},
\vspace{-5px}
\end{equation*}%
where $|E|$, $|V|$ and $|P|$ are the number of edges, vertices, and partitions , respectively.
\end{theorem}
\begin{proof}
We will prove the formula using a potential function as follows:
\begin{equation*}
	\Phi(t) := |E_{rest}(t)| + |V_{rest}(t)| + |P_{rest}(t)| + \sum_{\mathclap{p \in P}} |V(E_p(t))|,
\end{equation*}
where $t$ is the counter of the iteration in Algorithm~\ref{alg:parallelpartition}; $E_{rest}(t)$ is the set of non-allocated edges at $t$; $V_{rest}(t)$ is the set of vertices which has non-allocated adjacent edges at $t$; $P_{rest}(t)$ is the set of parallely computed partitions whose edges are less than the limit; $E_p(t)$ is the set of allocated edges; and $V(E_p(t))$ is the set of vertices covered by $E_p(t)$. 

Suppose Algorithm~\ref{alg:parallelpartition} is terminated at $T$, we will prove $\Phi_{T} \leq \Phi_{0}$.
Let $\Delta(\cdot)$ be the differentiation operator from $t$ to $t+1$ as $\Delta(\cdot) := (\cdot)(t+1) - (\cdot)(t)$.
Then, $\Delta\Phi$ is represented as follows:  
\begin{equation*}
  \Delta\Phi = \Delta|E_{rest}| + \Delta|V_{rest}| +  \Delta|P_{rest}| + \sum_{\mathclap{p \in P}} \Delta|V(E_p)| 
\end{equation*}
Suppose a vertex $x_p$ is selected from the boundary set~$B_p$ for each partition~$p \in P_{rest}(t)$ in parallel at $t$ as shown in Algorithm~\ref{alg:parallelpartition}.
Let $n_p$ be the number of $x_p$'s one-hop-neighbor edges added to $E_p(t)$, and $n'_p$ be the number of the two-hop-neighbor edges also added to $E_p(t)$.
Obviously, $n_p$ and $n'_p$ is equal to zero if $p$'s expansion has finished already.  
Then, $\Delta|E_{rest}|$ is represented as follow.
\begin{equation*}
  \Delta|E_{rest}| = - \sum_{\mathclap{p \in P}} (n_p + n'_p)
\end{equation*}
$|V_{rest}(t)|$ is a non-increasing function over $t$, and $\Delta|V_{rest}|$ can be represented using an indicator function $\chi(p)$ as follow.
\begin{equation*}
    \Delta|V_{rest}| \leq - \sum_{\mathclap{p \in P}} \chi(p), \quad
    \chi(p) = \begin{cases}
                0 & (p {\rm\ ends\ at\ } t {\rm\ or\ earlier}) \\
                1 & ({\rm Other})
             \end{cases}
\end{equation*}
This means that if $p$'s computation is terminated at $t$, then $x_p$'s one-hop neighbors may not fully allocated, and $|V_{rest}(t + 1)|$ at $p$ may be equal to $|V_{rest}(t)|$.
Otherwise, $x_p$'s one-hop neighbors are fully allocated, and $|V_{rest}(t)|$ at $p$ is always decreased by up to one ($x_p$ and the other vertices due to the allocation of the two-hop-neighbor edges).

$|P_{rest}(t)|$ is also a non-increasing function over $t$.
\vspace{-5px}
\begin{equation*}
    \Delta|P_{rest}| = - \sum_{\mathclap{p \in P}} (1 - \chi(p))
\vspace{-5px}
\end{equation*}
If $p$'s computation is terminated at $t$, then the number of computed partition $|P_{rest}(t)|$ is decreased ($\because \chi(p) = 0$); otherwise it is constant ($\because \chi(p) = 1$).

$|V(E_p(t))|$ is a increasing function over $t$, and $\Delta|V(E_p)|$ can be represented as $\Delta|V(E_p)| \leq n_p + 1$, where $n_p + 1$ means the number of $x_p$'s allocated neighbors plus $x_p$ itself.  

To sum up above, the difference of the potential function from $t$ to $t+1$ is represented as follow:
\begin{eqnarray*}
	\Delta\Phi &=& \Delta|E_{rest}| + \Delta|V_{rest}| + \Delta|P_{rest}| + \sum_{\mathclap{p \in P_{rest}(t)}} \Delta|V(E_p)| \\
	        &\leq& - \sum_{\mathclap{p \in P_{rest}(t)}} \{ (n_p + n'_p)  + \chi(p) +  (1 - \chi(p)) - (n_p + 1) \} \\
	           &=& - \sum_{\mathclap{p \in P_{rest}(t)}} n'_p \quad \leq 0
\end{eqnarray*}

Therefore, $\Delta\Phi \leq 0$ for all $t$, and thus $\Phi_{T} \leq \Phi_{0}$.
The following equation is established.
\begin{equation*}
    RF :=  \sum_{p \in P} \frac{|V(E_p)|}{|V|} = \frac{\Phi_{T}}{|V|} \leq \frac{\Phi_{0}}{|V|} = \frac{|E| + |V| + |P|}{|V|}
\end{equation*}
Thus, Theorem~\ref{thr:upperbound} is proved.
\end{proof}
In addition, the tightness of the upper bound (called \emph{UB}) is provided by the following theorem.
\begin{theorem}[Tightness]\label{thr:tightness}
\begin{equation*}
\exists G(V,E), \exists P \ s.t.\ RF \sim \mathit{UB}\ as\ |V| \to \infty.
\vspace{-5px}
\end{equation*}
\end{theorem}
\begin{proof}
Suppose a graph consisting of two isolated sub-graphs: a complete graph and a ring.
We assume that the complete graph includes $n$ vertices and $\frac{1}{2}n(n-1)$ edges, and that the ring has $\frac{1}{2}n(n-1)$ vertices and $\frac{1}{2}n(n-1)$ edges.
We will show the graph satisfies the condition of the theorem in a certain number of partitions.
Specifically, we will prove that for the graph, 
\begin{equation*}
|P| = 2^{-1}n(n-1) \Rightarrow RF \sim \mathit{UB}\ as\ |V| \to \infty.
\end{equation*}

Figure~\ref{fig:tightness} shows the example of the case $n = 4$.
In general, the total number of vertices $|V|$ is equal to $\frac{1}{2}n(n-1) + n$.
The total number of edges $|E|$ is $\frac{1}{2}n(n-1) + \frac{1}{2}n(n-1) = n(n-1)$.
\begin{figure}[h]
  \vspace{-5pt}
  \centering
  \includegraphics[width=.3\columnwidth]{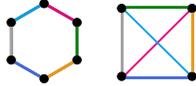}
  \vspace{-5pt}
  \caption{Partitions of Ring + Complete Graph.}
  \label{fig:tightness}
  \vspace{-5pt}
\end{figure}

Consider a case where the initial random vertex selection for each partition is from the ring, and each partition allocates one of the edges.
Then, in the next iteration, the random vertex selection occurs again in the complete graph since no partitions have boundary vertices.
Each partition allocates one of its one-hop-neighbor edges and then terminates the expansion because the number of allocated edges reaches to the limit ($\frac{|E|}{|P|} = \frac{n(n-1)}{2^{-1}n(n-1)} = 2$).
Finally, each edge in the complete graph is allocated to one of the partitions.
Figure~\ref{fig:tightness} shows the final condition, where different colors represent different partitions.
In each partition, the allocated edges are disconnected.

In this case, both the complete graph and the ring have $n(n-1)$ vertex replications. 
The replication factor is represented as follows:
\begin{equation*}
  RF = |V|^{-1}\{n(n-1) + n(n-1))\} = |V|^{-1}\{2n(n-1)\}
\end{equation*}

The upper bound is represented as follows:
\begin{eqnarray*}
  \mathit{UB} &:=& |V|^{-1}(|E| + |V| + |P|) \\
     &=&  |V|^{-1}\{n(n-1) + 2^{-1}n(n-1) + n + 2^{-1}n(n-1)\} \\
     &=&  |V|^{-1}\{2n(n-1) + n\}
\end{eqnarray*}
Therefore,
\begin{equation*}
    \lim_{|V| \to \infty} \frac{RF}{\mathit{UB}} = \lim_{n \to \infty} \frac{2n(n-1)}{2n(n-1) + n}  = 1 \ \ 
    \therefore RF \sim \mathit{UB} \qed
\end{equation*}%
\renewcommand{\qedsymbol}{}%
\end{proof}
\vspace{-5px}
\noindent\textbf{Comparison with the Other Distributed Methods.}
We compare the theoretical upper bound of Distributed NE with the other distributed edge partitioning methods discussed in~\cite{NIPS2014_5396}.
In the paper, the theoretical upper bounds of three hash-based methods (\textit{Random}, \textit{Grid}, and \textit{DBH}) are formulated in the case of power-law graphs.

Since our upper bound is for the general graph, we will apply it to the power-law graph. 
Based on the formulation of the power-law distribution by Clauset et al.~\cite{clauset2009power}, we use a graph, $G_{\zeta}(V,E)$, satisfying the following condition.
\begin{equation}
    Pr[d] = d^{-\alpha} \cdot \zeta(d,d_{min})^{-1}, \label{eq:degree}
\end{equation}
where $Pr[d]$ is the probability that vertex's degree becomes $d$; $\alpha$ is the scaling
parameter (Typically, $2 < \alpha < 3$); $\zeta(d,d_{min})$ is the generalized (or Hurwitz) zeta function; and $d_{min}$ is the minimum degree.
We assume that $d_{min} = 1$ (thus, $\zeta(d,d_{min})$ becomes Riemann zeta function) and that $|V|$ is much bigger than $|P|$ such that $|P| / |V| \approx 0$.

In $G_{\zeta}$, the expected upper bound of Distributed NE is represented as
\begin{equation*}
    \mathbb{E}[UB] \approx \mathbb{E}\left[\frac{|E|}{|V|}\right] + 1 = \frac{1}{2}\cdot\frac{\zeta(\alpha - 1, 1)}{\zeta(\alpha, 1)} + 1.
\end{equation*}

This is because the mean of \eqref{eq:degree} is equal to $\zeta(\alpha-1, 1)/\zeta(\alpha, 1)$ when $d_{min} = 1$, while the mean degree of a graph is generally $2 \times \mathbb{E}[|E| / |V|]$.

Table~\ref{tbl:upperbound} shows the calculation result in various $\alpha$ under $|P| = 256$.
We calculate the existing upper bound based on the formulas provided in~\cite{NIPS2014_5396}.
Distributed NE always provides the better upper bound than the existing distributed methods.
Especially, in the smaller $\alpha$, the difference is more significant.

\begin{table}[h]
\vspace{-13pt}
\centering
\caption{Theoretical Upper Bound of Replication Factor in Power-law Graph (256 Partition).}
\label{tbl:upperbound}
\scalebox{0.8}{
\begin{tabular}{|l|r|r|r|r|}
\hline
Partitioner & $\alpha = $ 2.2 & 2.4 & 2.6 & 2.8 \\ \hlinewd{2\arrayrulewidth}
\textit{Random} (1D-hash)         & 5.88 & 3.46 & 2.64 & 2.23 \\
\textit{Grid} (2D-hash) & 4.82    & 3.13 & 2.47 & 2.13 \\
\textit{DBH}~\cite{NIPS2014_5396} & 5.54 & 3.19 & 2.42 & 2.05 \\
\textbf{\textit{Distributed NE}}  & \textbf{2.88} & \textbf{2.12} & \textbf{1.88} & \textbf{1.75} \\ \hline
\end{tabular}
}
\end{table}

\begin{highlight}
We provide the time complexity for each computing unit since Distributed NE is executed in parallel.
Let $n$ be the number of the computing units (i.e., the cores in practice).
\vspace{-5px}
\begin{theorem}[Efficiency]\label{thr:time}
In Distributed NE, suppose the initial graph is evenly distributed to each allocation process; the workload in each machine is evenly assigned to its computing units; and $|V|,|E| \gg |P|, n$.
The worst-case time complexity of the local computation per unit is $\mathcal{O}\Bigl(\frac{d|E|(|P| + d)}{n|P|}\Bigr)$, where $d$ is the maximum degree.
\end{theorem}
\vspace{-10pt}
\begin{proof}
The dominant part in the local computation is \\ \texttt{AllocateTwoHopNeighbors()} in Algorithm~\ref{alg:details}.
In the worst case, the number of messages in $BP_{new}$ is $\mathcal{O}(|E| / |P|)$.
In Line 12, $|BP_{new}| / n = \mathcal{O}(|E| / n |P|)$ vertices per unit are processed.
For each vertex $u$, $\mathcal{O}(d)$ neighbors are processed.
The complexity from Line~13 to 18 is $\mathcal{O}(|P| + d)$ since the time complexity is $\mathcal{O}(|P|)$ at Line 14,16 and $\mathcal{O}(d)$ at Line 17.
Therefore, the total time complexity of the function is $\mathcal{O}\Bigl(\frac{|E|}{n |P|}\Bigr) \times \mathcal{O}(d) \times \mathcal{O}(|P| + d) = \mathcal{O}\Bigl(\frac{d|E|(|P| + d)}{n|P|}\Bigr)$
\end{proof}
\vspace{-8pt}
\end{highlight}

\section{Empirical Evaluation} \label{sec:evaluation}

\begin{figure*}[h]
 \centering
 \subfigure[\texttt{Pokec}]{\includegraphics[width=.2\textwidth]{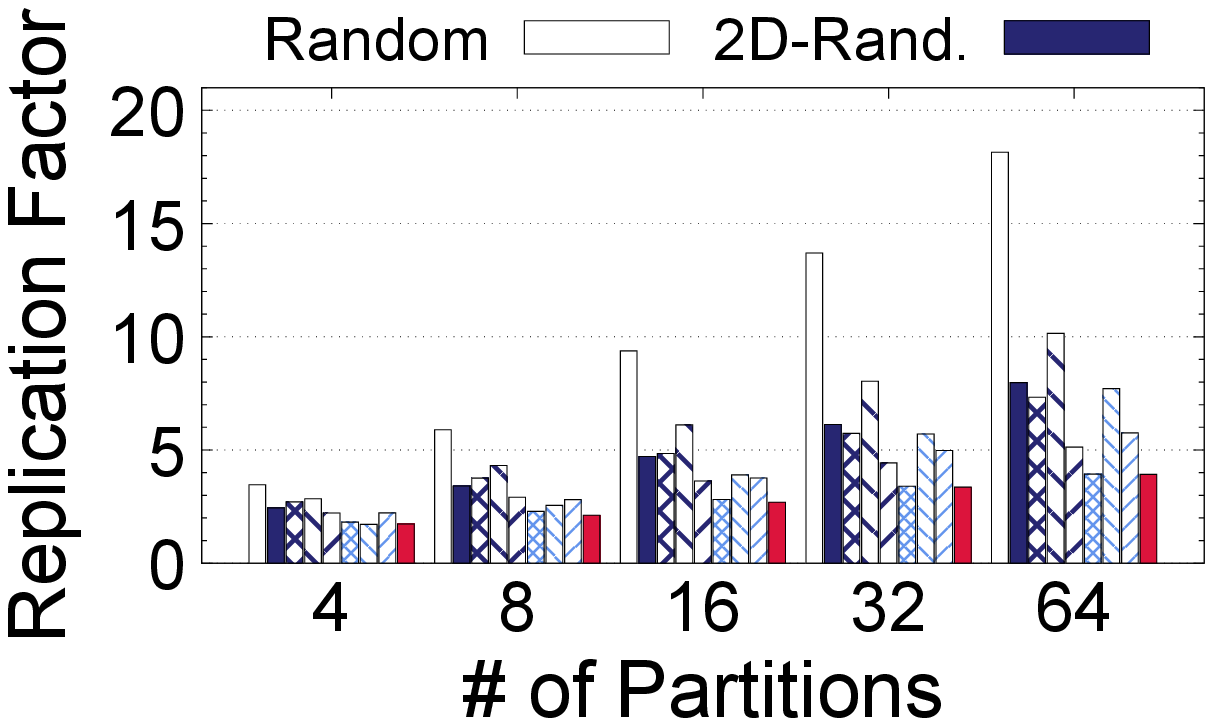}\label{fig:quality-pockec}}%
 \subfigure[\texttt{Flickr}]{\includegraphics[width=.2\textwidth]{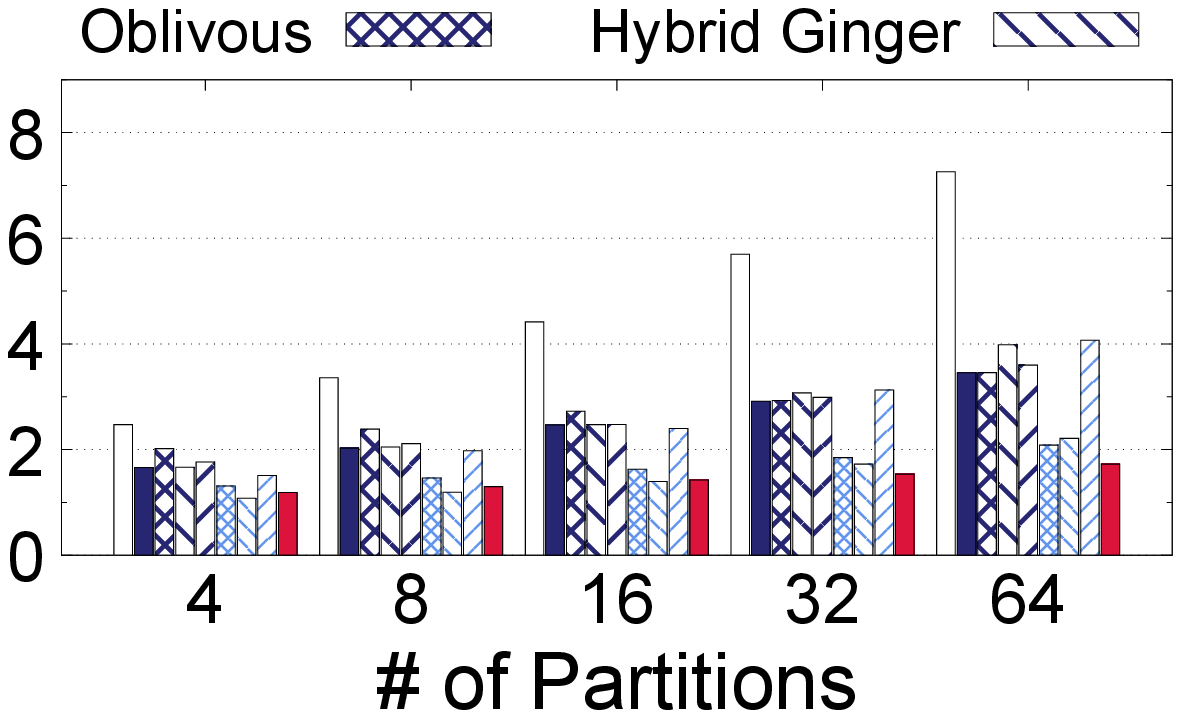}\label{fig:quality-flicker}}%
 \subfigure[\texttt{LiveJ.}]{\includegraphics[width=.2\textwidth]{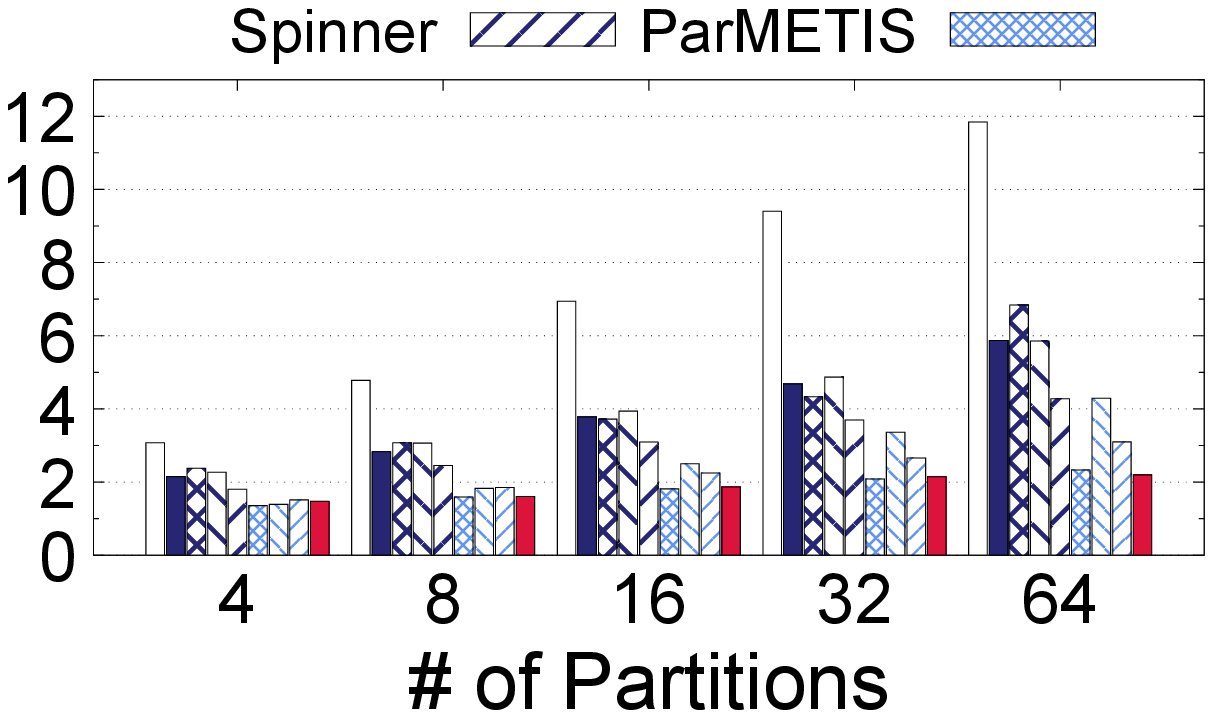}\label{fig:quality-livejournal}}%
 \subfigure[\texttt{Orkut}]{\includegraphics[width=.2\textwidth]{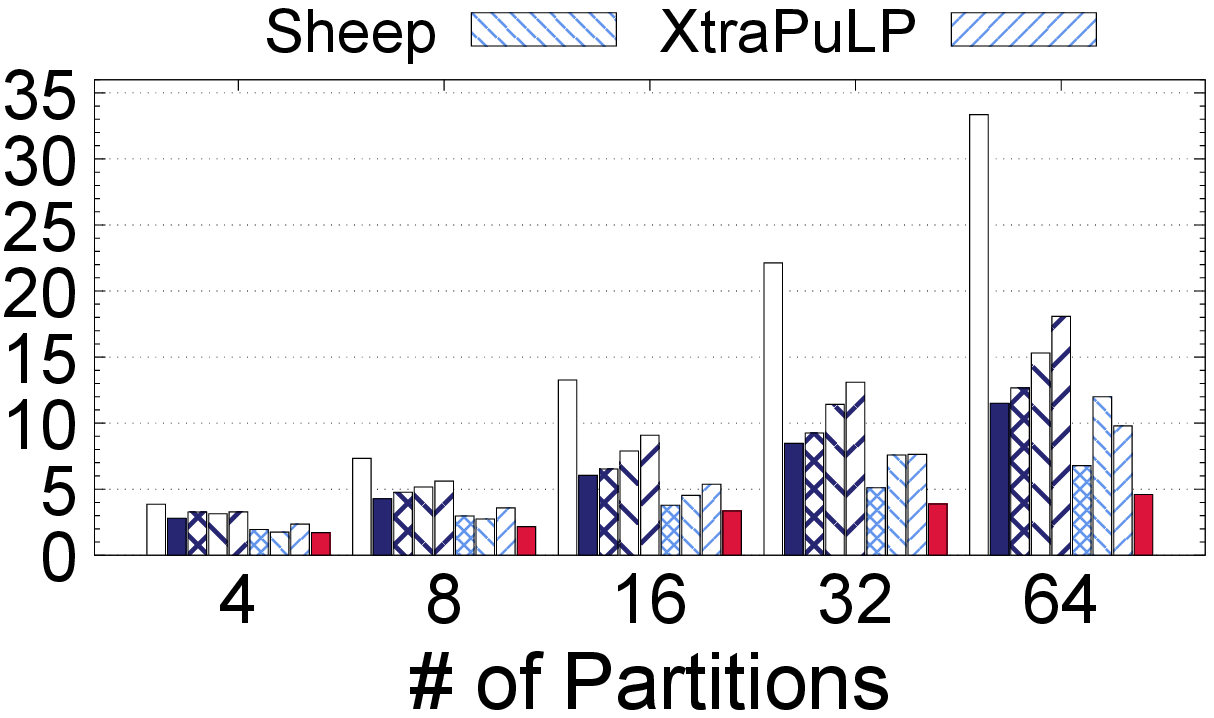}\label{fig:quality-orkut}}%
 \subfigure[\texttt{Twitter}]{\includegraphics[width=.2\textwidth]{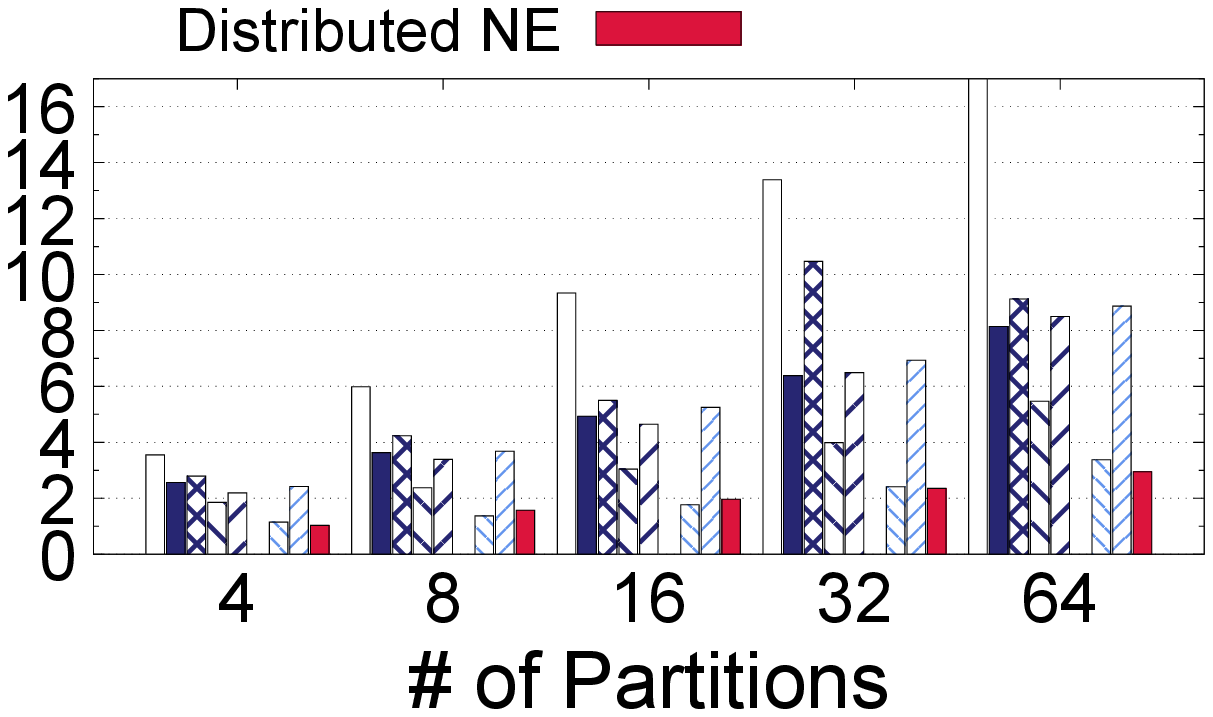}\label{fig:quality-twitter}}\\
 \vspace{-5pt}
 \subfigure[\texttt{Friendster}]{\includegraphics[width=.2\textwidth]{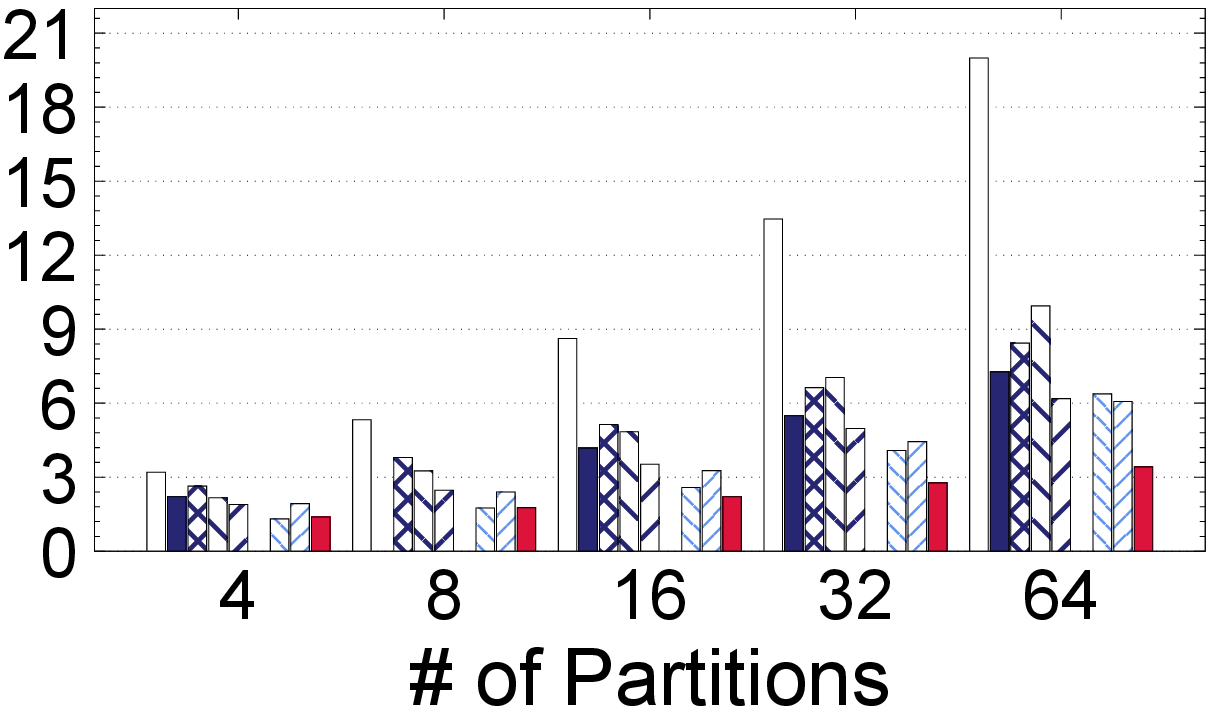}\label{fig:quality-friendster}}%
 \subfigure[\texttt{WebUK}]{\includegraphics[width=.2\textwidth]{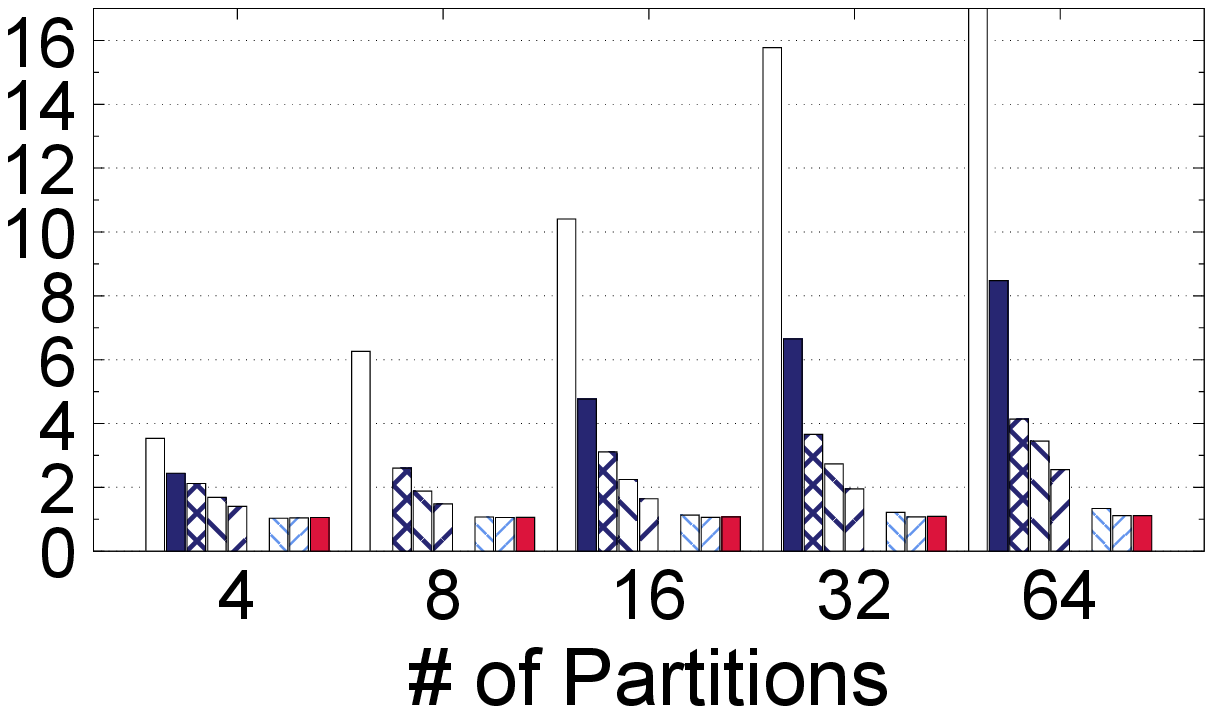}\label{fig:quality-webuk}}%
 \subfigure[\texttt{RMAT Scale20}]{\includegraphics[width=.2\textwidth]{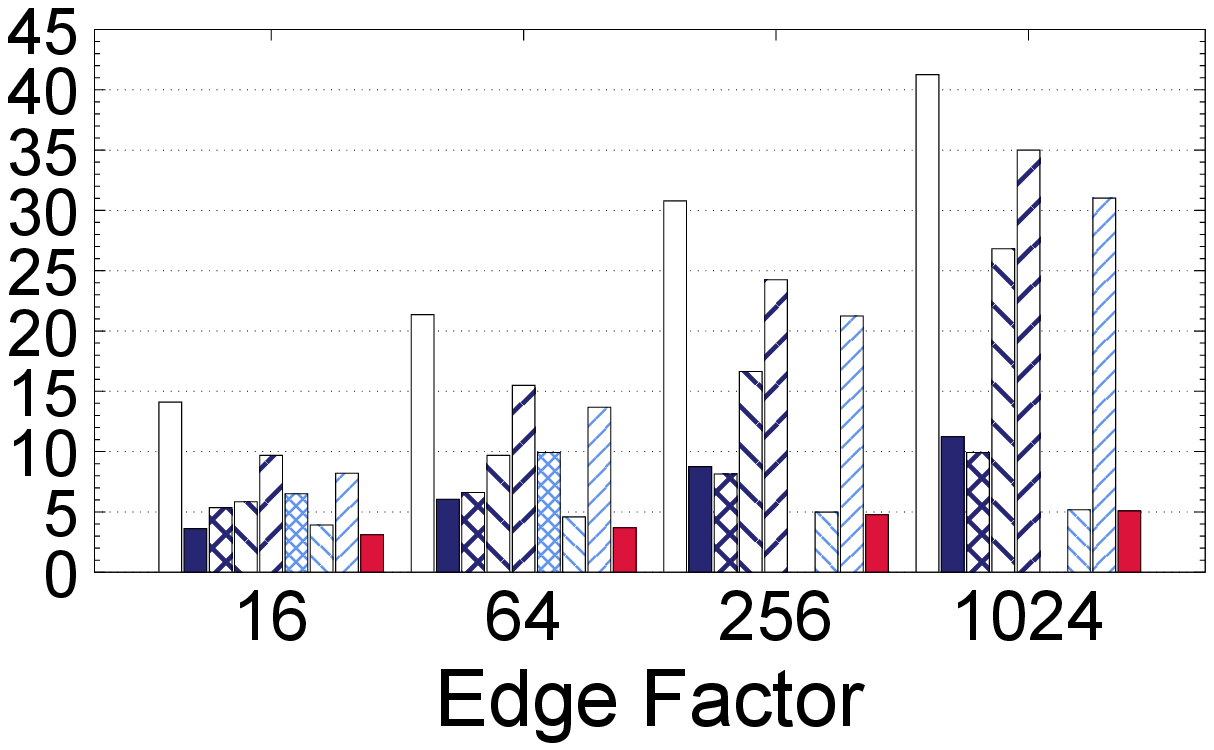}\label{fig:quality-rmat1}}%
 \subfigure[\texttt{RMAT Scale21}]{\includegraphics[width=.2\textwidth]{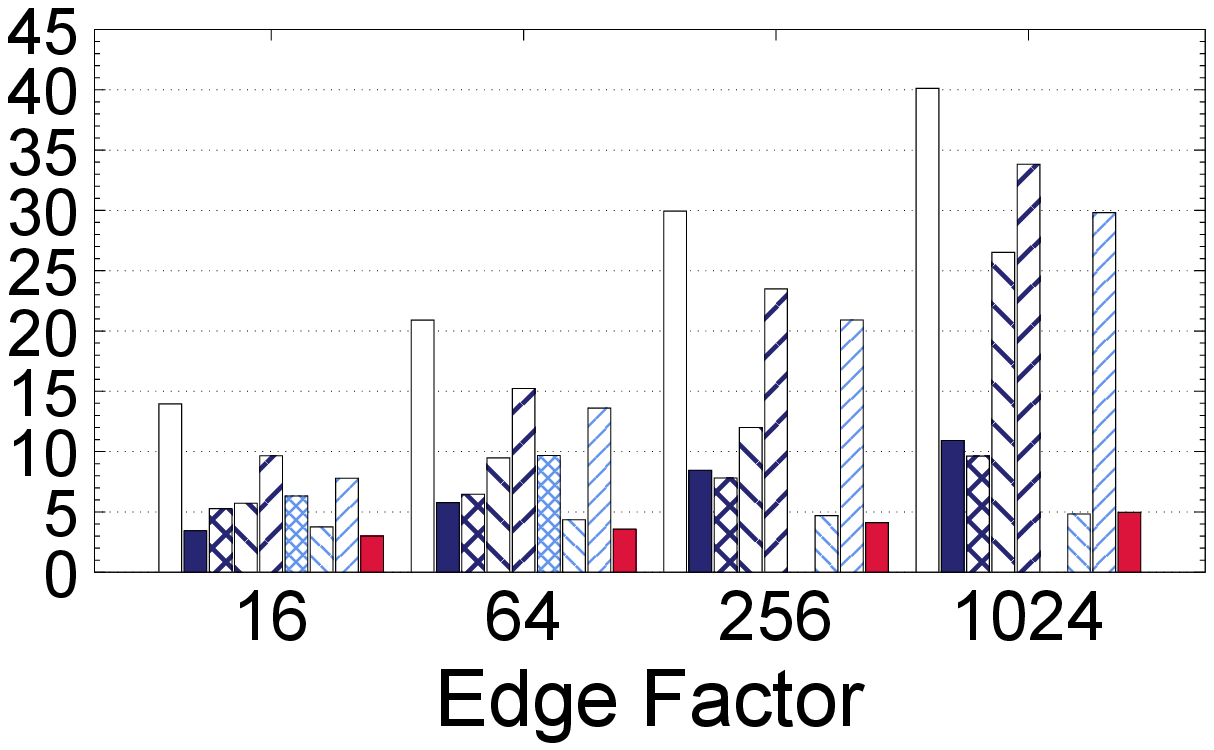}\label{fig:quality-rmat2}}%
 \subfigure[\texttt{RMAT Scale22}]{\includegraphics[width=.2\textwidth]{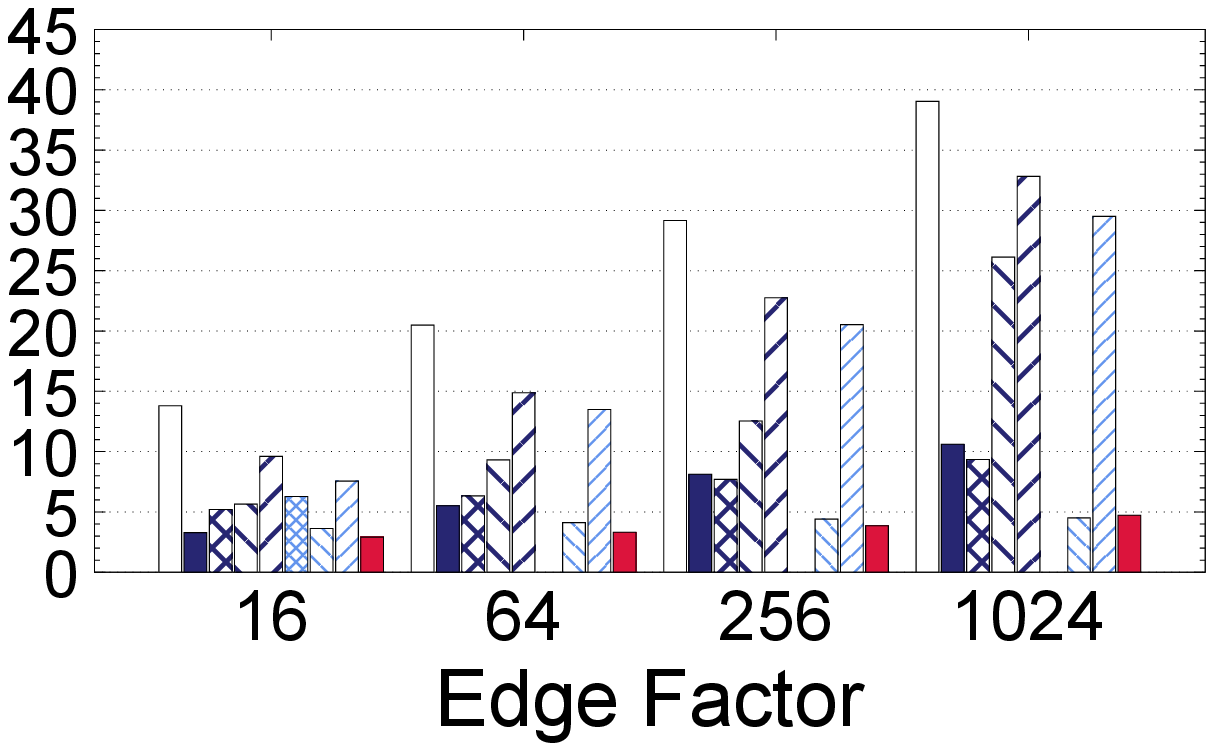}\label{fig:quality-rmat3}}%
 \vspace{-10pt}
 \caption{Replication Factor of Real-world Graphs and RMAT Graphs. $|P| = 64$ in RMAT Graphs.}\label{fig:quality}
 \vspace{-15pt}
\end{figure*}

In this section, we discuss the empirical analysis of the quality, scalability, and efficiency.
Our claims are as follows:

\smallskip
\noindent\textbf{Highest Quality.} In various types of skewed graphs, Distributed NE always generates higher-quality partitions than any other state-of-the-art distributed methods.

\noindent\textbf{Highest Scalability.} Memory usage of Distributed NE is an order of magnitude smaller than that of the state-of-the-art high-quality distributed methods. It can handle a larger graph with fewer machines than the existing methods.

\noindent\textbf{Comparable Efficiency.} The elapsed time of Distributed NE is comparable to or better than the state-of-the-art high-quality distributed methods.

\noindent\textbf{Trillion-edge Graph.} Due to its scalability and efficiency, Distributed NE can generate high-quality partitions of the trillion-edge graph using only a few hundreds of machines.




\subsection{Benchmarks and Setup}
\noindent\textbf{Real-world Datasets.}
For the evaluation, we use different graph datasets with over a billion of vertices, provided by SNAP~\cite{snapnets}, KONECT~\cite{KONECT,Kunegis:2013:KKN:2487788.2488173} and LWA~\cite{LWA}.
We use 7 real-world skewed graphs as summarized in Table \ref{tbl:real_world}.
\begin{table}[h]
\vspace{-13pt}
\centering
\caption{Real-world Skewed Graphs}
\label{tbl:real_world}
\scalebox{0.8}{
\begin{tabular}{|l|r|r|l|r|r|}
\hline
Dataset                & Vertices & Edges & Dataset & Vertices & Edges \\ \hlinewd{2\arrayrulewidth}
\texttt{Pokec} \cite{takac2012data} & 1.63M & 30.62M & \texttt{Twitter} \cite{Kwak:2010:TSN:1772690.1772751} & 41.65M  & 1.46B \\
\texttt{Flickr} \cite{Mislove:2008:GFS:1397735.1397742}       & 2.30M      & 33.14M & \texttt{Friendster} \cite{yang2012defining} & 65.60M & 1.80B          \\ 
\texttt{LiveJ.} \cite{Backstrom:2006:GFL:1150402.1150412}   & 4.84M      & 68.47M  & \texttt{WebUK} \cite{BSVLTAG} & 105.15M & 3.72B   \\ 
\texttt{Orkut} \cite{6413740} & 3.07M      & 117.18M & & &      \\ \hline
\end{tabular}
}
\vspace{-10pt}
\end{table}


\smallskip
\noindent\textbf{Synthetic Datasets and Trillion-edge Graph.}
Since each real-world dataset has the entirely different graph feature, the results would not provide any common performance characteristics.
Thus, further experiments are conducted using various sizes of synthetically generated graphs which have the similar graph features.
We use RMAT~\cite{chakrabarti2004r,Leskovec2010} graphs whose vertex size are from Scale20 to Scale30, where \emph{Scale$N$} is referred to as a graph with $2^N$ vertices.
Their average number of edges per vertex, called \emph{edge factor} (\emph{EF}), ranges from $2^4$ to $2^{10}$ according to Graph500 setting~\cite{graph500} (the edge factor is $2^4$) and Facebook's trillion-edge graph~\cite{ching2015one}, which is reported to have 1.45 billions vertices and 1 trillion edges (thus, the edge factor is around $2^{10}$).

The trillion-edge graph is simulated using RMAT with the same scale as Facebook's trillion-edge graph because no real-world trillion-edge graph is publicly available. 
The graph is Scale30 with Edge factor $2^{10}$, namely, the graph consists of $1,073,741,824$ ($= 2^{30}$) vertices and $1,099,511,627,776$ ($= 2^{30} \times 2^{10}$) edges.


\noindent\textbf{Benchmark Partitioning Algorithms.}
We compare \textit{Distributed NE} with \begin{highlight}8\end{highlight} distributed partitioning methods.
\textit{Random} is the simple one-dimensional hashing.
\begin{highlight}\textit{2D-Random} (or referred to as \textit{Grid}) is the two-dimensional hashing, which is used for the initial assignment of \textit{Distributed NE}. \end{highlight}
\textit{Oblivious}~\cite{joseph2012powergraph} and \textit{Hybrid Ginger}~\cite{Chen:2015:PDG:2741948.2741970} are the state-of-the-art hash-based edge partitioning methods including iterative refinements.
\textit{Spinner}~\cite{7930049} is the state-of-the-art hash-based vertex partitioning method, where vertices are assigned randomly followed by the iterative refinements based on Label Propagation.
\textit{ParMETIS}~\cite{Karypis:1998:PAM:287098.287108} is the standard multi-level vertex partitioning.
\textit{Sheep}~\cite{Margo:2015:SDG:2824032.2824046} is the state-of-the-art high-quality distributed edge partitioning method based on the elimination-tree conversion.
\textit{XtraPuLP}~\cite{7967155} is the state-of-the-art high-quality distributed vertex partitioning method, where vertices are directly assigned based on Label Propagation without initial random allocation.
For the comparison of the partitioning quality with the vertex partitioning methods, such as \textit{ParMETIS}, \textit{Spinner}, and \textit{XtraPuLP}, we convert their vertex-partitioned graphs into the edge-partitioned ones as demonstrated in~\cite{Bourse:2014:BGE:2623330.2623660}, that is, each edge is randomly assigned to one of its adjacent vertices' partitions.

\noindent\textbf{Parameter Setting.}
The imbalance factor $\alpha$ discussed in Section~\ref{sec:background} is set to $1.1$.
Moreover, we set the expansion factor $\lambda$ equal to $0.1$ to balance the quality and performance as discussed in Section~\ref{sec:multiex}.

\vspace{-5pt}
\subsection{Quality Evaluation}
The quality of partitioning is evaluated by means of the replication factor (RF) as Equation~\eqref{eq:replicationfactor} in Section~\ref{sec:background}.
We execute \textit{Distributed NE} with five different random seeds and show the median value, where the relative standard error of the result is less than 5\%.

Figure~\ref{fig:quality-pockec} to \ref{fig:quality-webuk} show the replication factor of real-world graphs on different numbers of partitions from 4 to 64.
Note that even on our big memory server (around 1TB), \textit{ParMETIS} is unable to process \texttt{Twitter}, \texttt{Friendster}, and \texttt{WebUK} due to the insufficient memory space.

Overall, \textit{Distributed NE} outperforms the other methods in real-world graphs.
Especially in the severe cases, such as, a case of more partitions or a case where replication factor is relatively high (e.g., \texttt{Orkut}, \texttt{Pokec}, and \texttt{Friendster}), the improvement from the others is much more significant.
Only in the smaller number of partitions or in graphs whose replication factor is low, the other methods are comparable to \textit{Distributed NE}. 
For example, in \texttt{Flickr} and \texttt{Twitter} of 4 to 16 partitions, \textit{Sheep} is slightly better than \textit{Distributed NE}. 
In \texttt{LiveJ.}, \textit{ParMETIS} is also slightly better than \textit{Distributed NE} in 4 to 16 partitions.
In \texttt{WebUK}, replication factor of \textit{Distributed NE} is similar quality to that of \textit{Sheep} and \textit{XtraPuLP}.
Their scores are less than 1.1, which is a nearly ideal case, that is, there are no vertex replications, and replication factor is equal to 1.

Moreover, the partitioning quality of \textit{Distributed NE} is stable.
It always provides high-quality partitions.
On the other hand, the other methods have some specific graphs which are unsuited to partition.
For example, \textit{Sheep} generates high-quality partitions in \texttt{Twitter} but worse in \texttt{Pokec},  \texttt{LiveJ.}, \texttt{Orkut} and \texttt{Friendster}.
\textit{XtraPuLP} is significantly worse in \texttt{Twitter}, and \texttt{Friendster} and RMAT graphs discussed below.

Figure~\ref{fig:quality-rmat1} to \ref{fig:quality-rmat3} show the replication factor of RMAT synthetic graphs on 64 partitions.
Note that \textit{ParMETIS} is unable to process  256 and 1024 in all scales due to the insufficient memory space.

As with the real-world graphs, \textit{Distributed NE} outperforms the other methods.
In general, the replication factor becomes higher as the increase of the edge factor.
The result is intuitive since graphs become more complicated as the increase of edges, and thus it becomes more difficult to generate high-quality partitions.
Moreover, in the same edge factor, the replication factor is almost the same in the different scales: Scale20, Scale21, and Scale22.
This means that the difficulty in partitioning a graph depends on its complexity rather than its scale. 

\begin{highlight}
To summarize, we can classify all the methods into 3 categories.
The first categories are based on the hashing, such as \textit{Random}, \textit{2D-Random},  \textit{Oblivious}, \textit{Hybrid Ginger} and \textit{Spinner}. 
These methods provide low quality due to the randomness.
The methods of the second category indirectly solve the graph partitioning problem by converting it into the other similar problem, such as label propagation (\textit{XtraPuLP}) or tree partitioning (\textit{Sheep}).
Such indirect methods provide high-quality partitions only for some cases.
Finally, the methods of the last category directly solve the optimization problem based on the approximate algorithm, such as \textit{ParMETIS} and \textit{Distributed NE}.
These methods stably generate high-quality partitions.
\end{highlight}


\vspace{-10pt}
\subsection{Performance Evaluation}
The performance of the partitioning algorithms is evaluated.
We compare the memory consumption and elapsed time of \textit{Distribute NE} with the three high-quality partitioning methods: \textit{ParMETIS}, \textit{Sheep}, and \textit{XtraPuLP}.
Hash-based algorithms, such as \textit{Random}, \textit{Oblivious}, \textit{Hybrid Ginger}, and \textit{Spinner}, are omitted from the performance evaluation because the partitions by these algorithms do not reach to the sufficient quality as shown previously.
Instead, these algorithms is efficient and scalable since they only include light-weight hash calculation and local refinements.

Table~\ref{tbl:conf} shows the configuration of the distributed environment used to evaluate the performance.
According to Facebook's work~\cite{ching2015one}, where trillion-edge graph can be processed on only 200 machines, we use the similar scale of the distributed environment consisting of up to $256$ machines.
\begin{table}[t]
\begin{center}
\caption{Computational Environment}
\label{tbl:conf}
\scalebox{0.8}{
\begin{tabular}[c]{lr} \hline
  Service                & ASPIRE 1 in NSCC Singapore \\
  CPUs per Machine      & Dual Sockets Intel E5-2690v3 (2 $\times$ 12 cores) \\
  Memory / Network      & 98 GB per Machine / InfiniBand EDR \\
  \# of Machines        & $4$ to $256$ ($2^{2}$ to $2^{8}$) \\ 
  C++ Compiler / MPI    & GCC 5.4 (--O3 flag) / IntelMPI 5.1.2 \\
  OS                    & Red Hat Server release 6.9 \\ \hline
\end{tabular}
}
\end{center}
\vspace{-8pt}
\end{table}





\noindent\textbf{Memory Consumption.}
The scalability of the graph partitioning is evaluated with memory consumption.
Figure~\ref{fig:memusage_real} shows the memory consumption with the real-world graphs on 64 machines (64 processes).
We take snapshots of all distributed processes every $0.5$ second during execution to get their memory usage.
Then, we use the snapshot, $s_{max}$, at which the total memory usage of the processes becomes maximum. 
The score is normalized by the number of edges, that is,
\vspace{-5pt}
\begin{equation*}
\textit{(Mem Score)} :=  \frac{1}{|E|} \sum_{pr \in Pr}\{pr\text{'s}\ \text{Mem.}\ \text{Usage}\ \text{(byte)}\ \text{at}\ s_{max}\},
\vspace{-5pt}
\end{equation*}
where $pr$ is a process, and $Pr$ is a set of the processes.
Overall, \textit{Distributed NE} outperform the other methods by around one order of magnitude.
On average, its mem score is only 5.89 \% of the other methods.
The mem score slightly decreases as the increase of the graph size because the proportion of the edges to the total memory usage becomes decrease accordingly.

\begin{figure}[h]
\vspace{-10pt}
 \centering
 \subfigure[Real-world Graphs]{\includegraphics[width=.5\columnwidth]{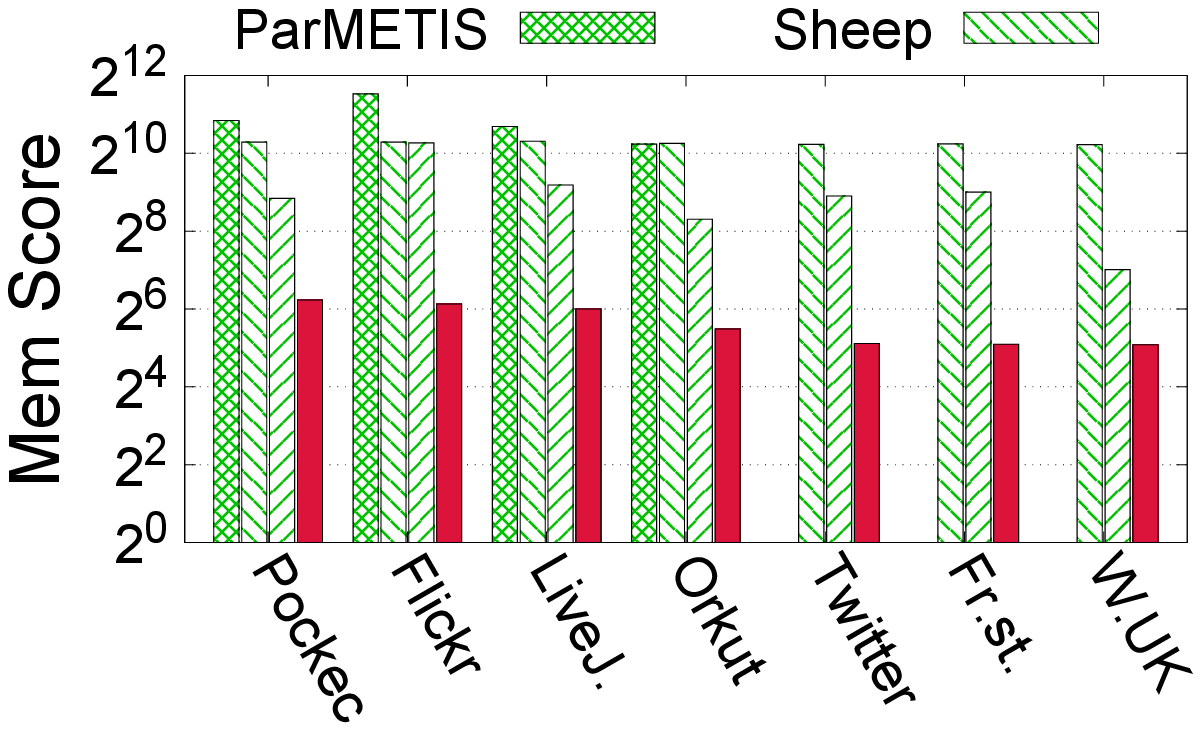}\label{fig:memusage_real}}%
 \subfigure[RMAT Graphs]{\includegraphics[width=.5\columnwidth]{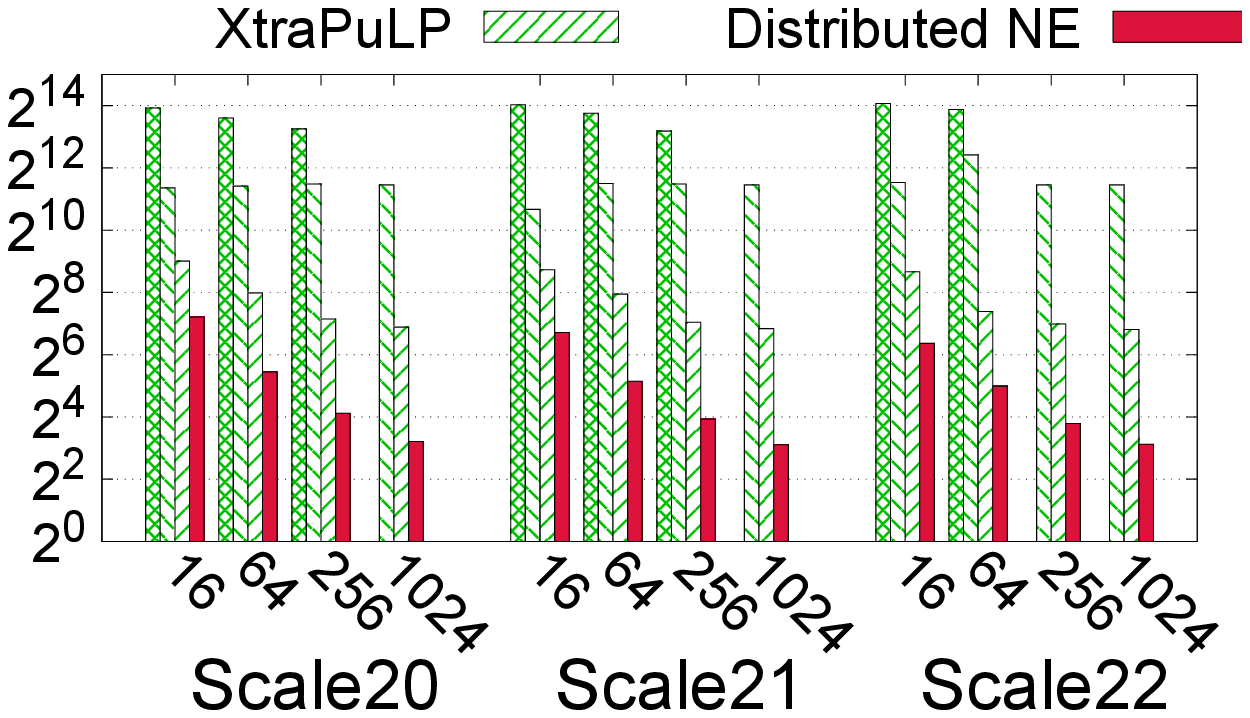}\label{fig:memusage_rmat}}
  \vspace{-13pt}
 \caption{Memory Consumption. \emph{Mem Score} is a total memory usage (byte) normalized by \# edges.}
  \vspace{-10pt}
\end{figure}

\textit{Distributed NE} is highly space-efficient because of two main reasons.
First, each distributed edge is unique without any replication among machines, while vertices are replicated.
In general, the replication of vertices is more space-efficient than that of edges since the vertex requires fewer bytes than the edge.
In the vertex partitioning such as \textit{ParMETIS} and \textit{XtraPuLP}, the memory consumption may become higher since the edges are replicated among machines instead of vertices. 
Especially in the skewed graph, the vertex replication is better since there are much more edges than vertices.
Moreover, in the common multiple coarsening-refinement approaches such as \textit{ParMETIS}, graph data are replicated multiple times for coarsening, and it requires much more memory than the others.
Second, in \textit{Distributed NE}, the graph data are stored without any memory-consuming data structure such as the hash map, which usually consumes around an order of magnitude memory space compared to the continuous array. 
The core components of the graph are stored in CSR, and their metadata is functionally computed instead of storing them.

Figure~\ref{fig:memusage_rmat} shows the memory consumption of the RMAT graphs on 64 machines.
\textit{ParMETIS} is unable to run all scales with  1024 and Scale22 with 256 due to out-of-memory.
As with the real-world graphs, \textit{Distributed NE} is much better than the others.
In \textit{Distributed NE}, the mem score shows the substantial decrease as the increase of the edge factor because during computation it compacts the duplicated edges, which have the same sources and destinations.
The duplication often appears in the higher edge factor.

\smallskip
\noindent\textbf{Elapsed Time.}
The efficiency is evaluated with the elapsed time.
We measure the time to compute partitions excluding the loading time to deploy input graph data.
The conversion time from the vertex partition to the edge partition is also excluded in \textit{XtraPuLP} and \textit{ParMETIS}.
We run the programs ten times and show the median value.

Figure~\ref{fig:performance-pokec} to \ref{fig:performance-webuk} show the elapsed time in the real-world graphs.
Overall, \textit{Distributed NE} outperform \textit{ParMETIS} and \textit{Sheep}.
In 64 partition, its speed up over \textit{ParMETIS} is up to 9.1 in \texttt{LiveJ.}.
The speed up over \textit{Sheep} is up to 19.8 in \texttt{Twitter}.
Its performance is basically comparable to \textit{XtraPuLP} and slightly worse in \texttt{Flickr} on 64 machines and \texttt{Friendster}.
In \texttt{Flickr} on 64 machines, \textit{Distributed NE} takes lots of iterations at the final part of the computation because many isolated edges remain to be allocated.
In the situation, edges are allocated by the random selection, but the number of allocated edges per iteration is few.
\texttt{Friendster} is the special case for \textit{XtraPuLP} to execute very fast as mentioned in the original paper~\cite{7967155}.

\begin{figure*}
 \centering
 \subfigure[\texttt{Pokec}]{\includegraphics[width=.20\textwidth]{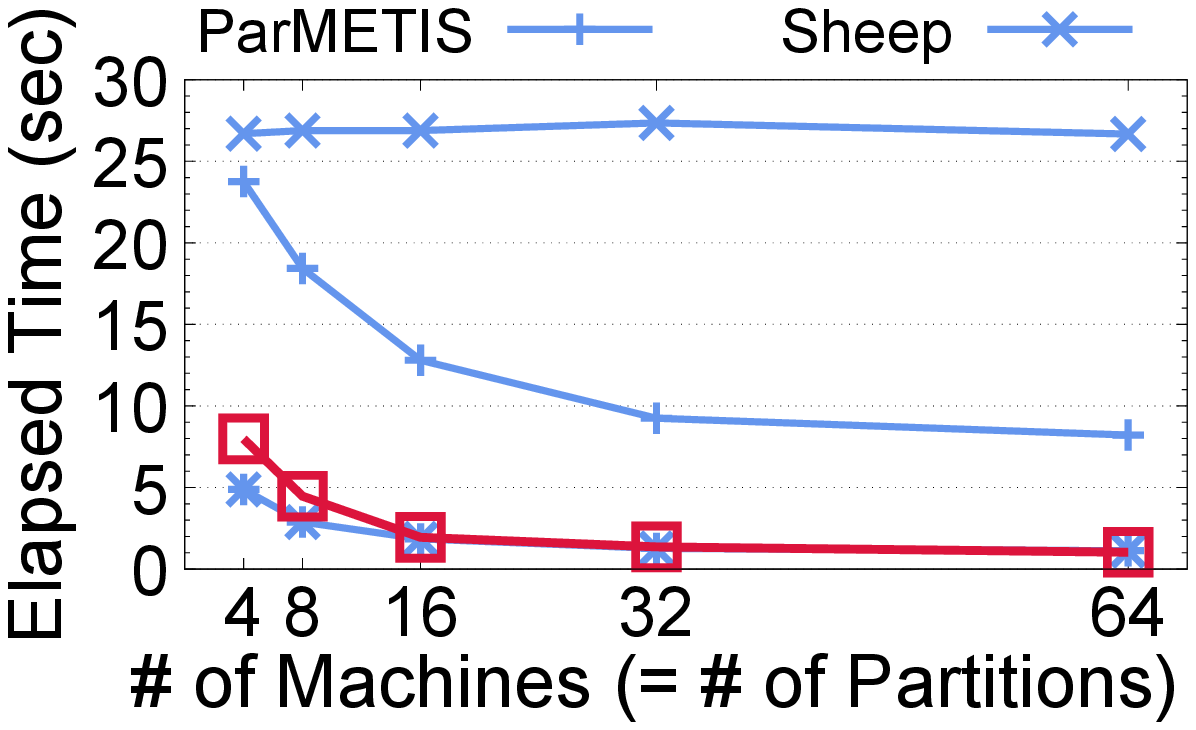}\label{fig:performance-pokec}}%
 \subfigure[\texttt{Flickr}]{\includegraphics[width=.20\textwidth]{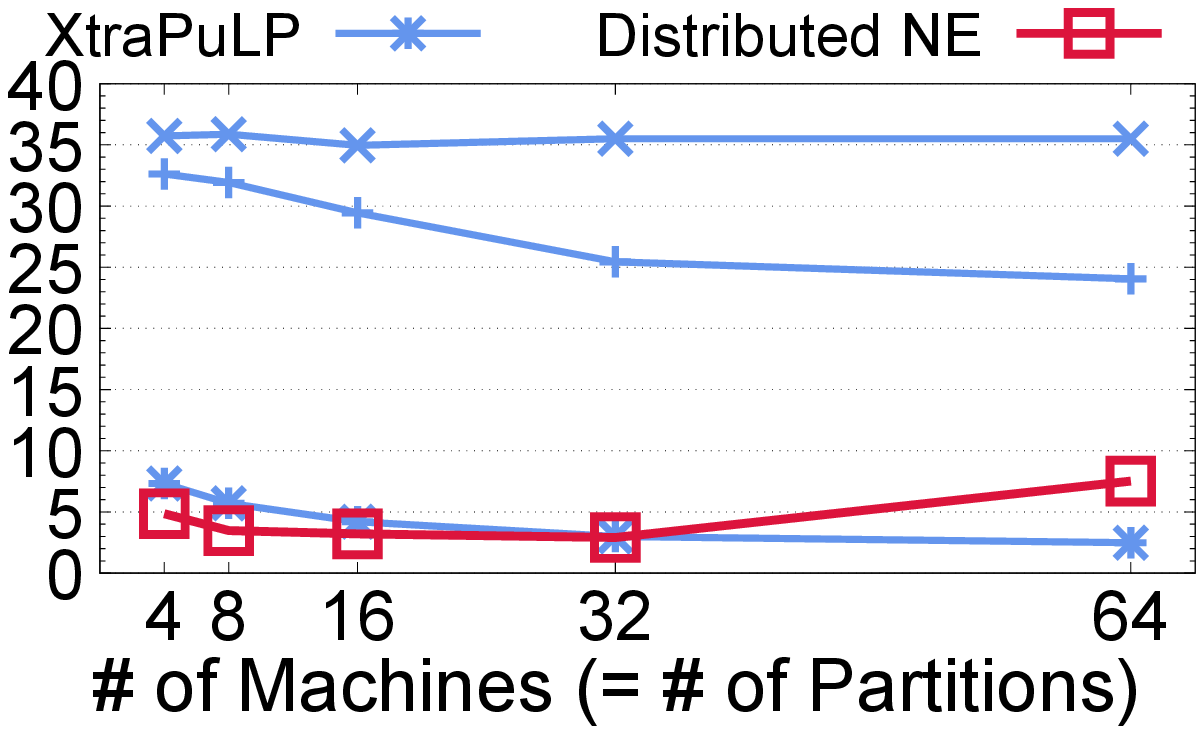}\label{fig:performance-flicker}}%
 \subfigure[\texttt{LiveJ.}]{\includegraphics[width=.20\textwidth]{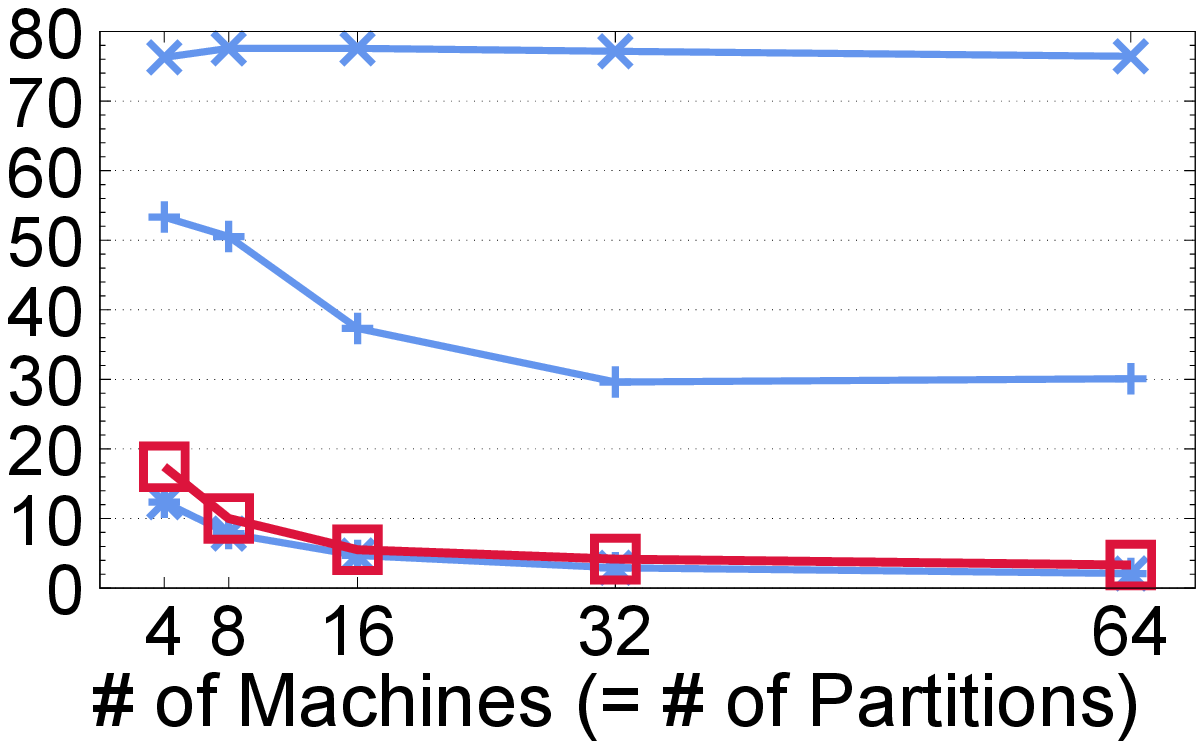}\label{fig:performance-livejournal}}%
 \subfigure[\texttt{Orkut}]{\includegraphics[width=.20\textwidth]{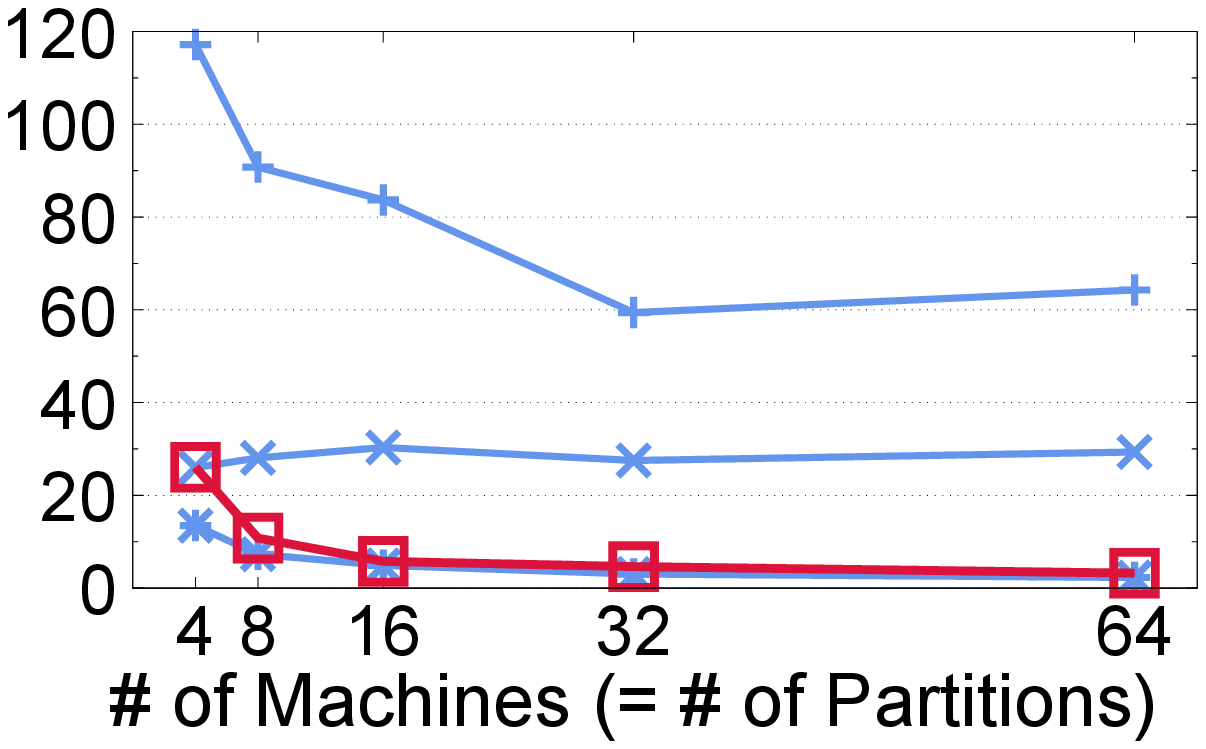}\label{fig:performance-orkut}}%
 \subfigure[\texttt{Twitter}]{\includegraphics[width=.20\textwidth]{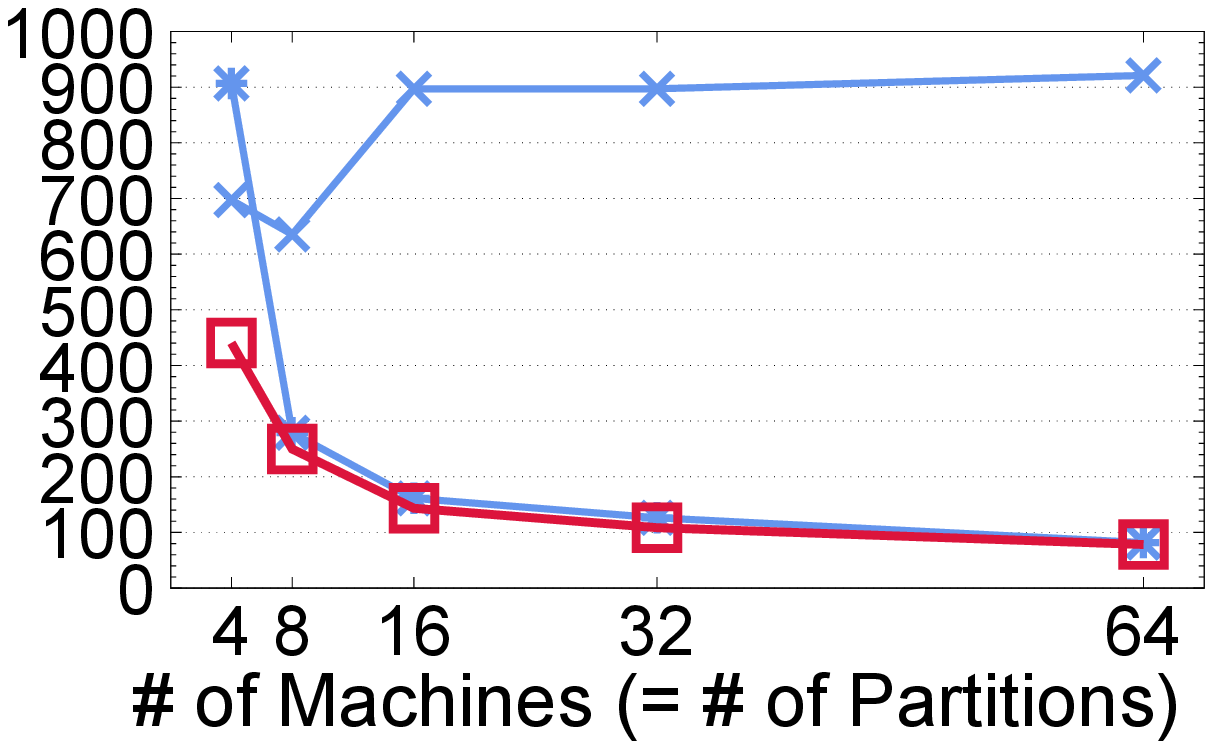}\label{fig:performance-twitter}} \\
 \subfigure[\texttt{Friendster}]{\includegraphics[width=.20\textwidth]{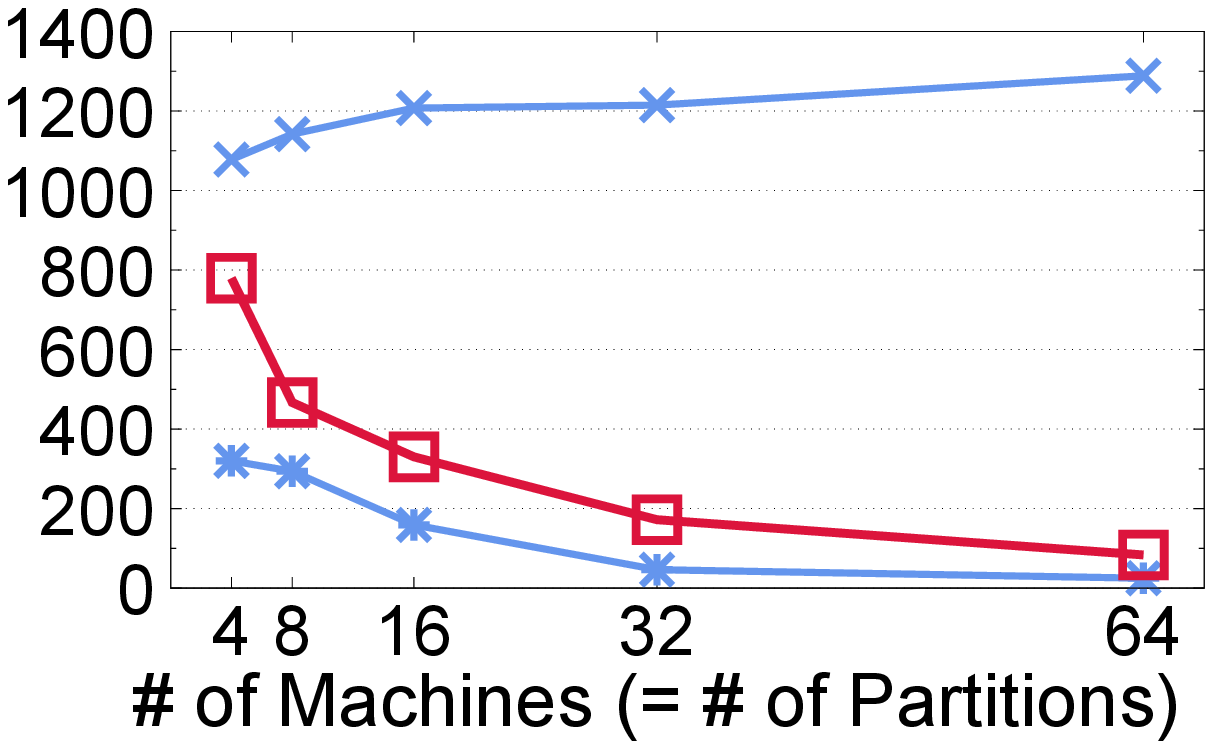}\label{fig:performance-friendster}}%
 \subfigure[\texttt{WebUK}]{\includegraphics[width=.20\textwidth]{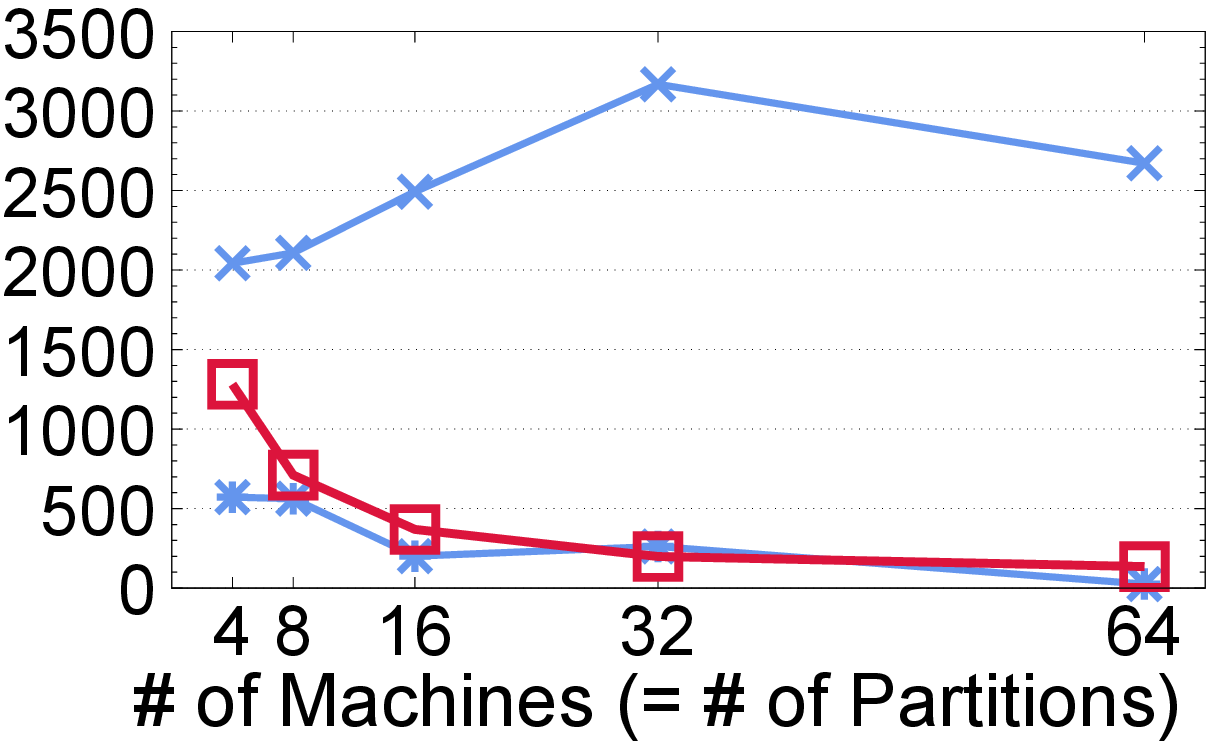}\label{fig:performance-webuk}}%
 \subfigure[\texttt{RMAT} Different EFs]{\includegraphics[width=.20\textwidth]{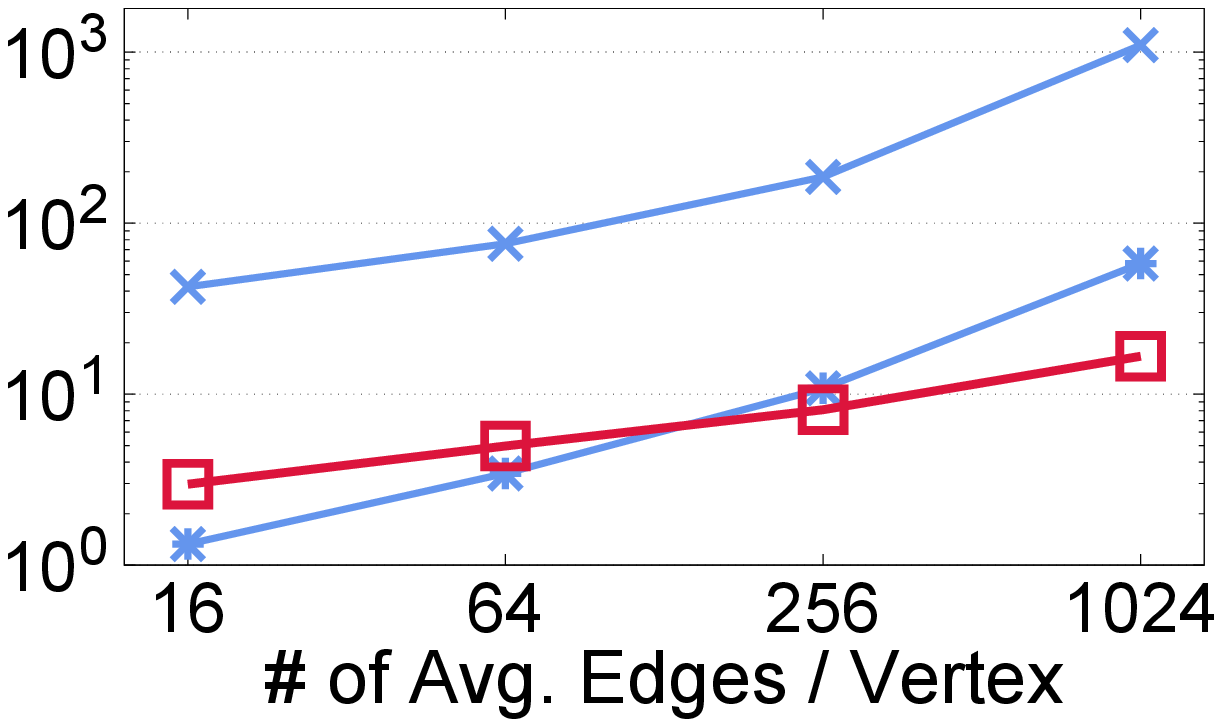}\label{fig:diff-edge}}%
 \subfigure[\texttt{RMAT} Different Scales]{\includegraphics[width=.20\textwidth]{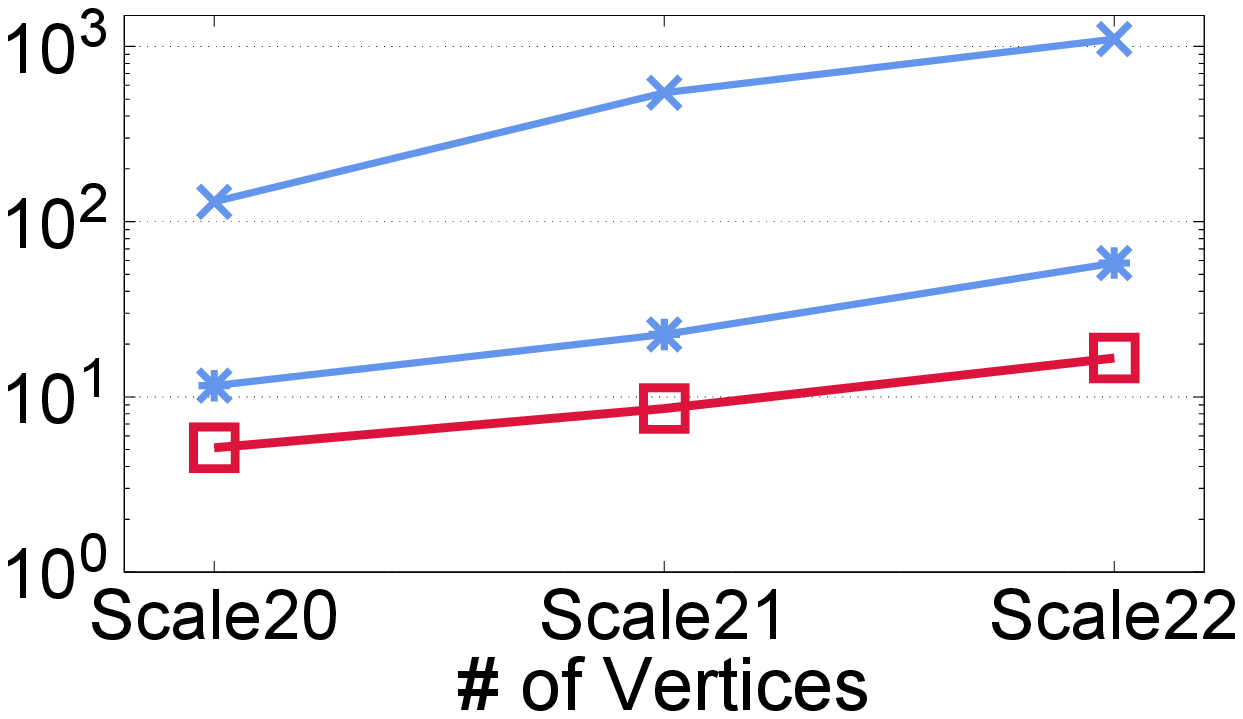}\label{fig:diff-vert}}%
 \subfigure[\texttt{Trillion Edges}]{\includegraphics[width=.20\textwidth]{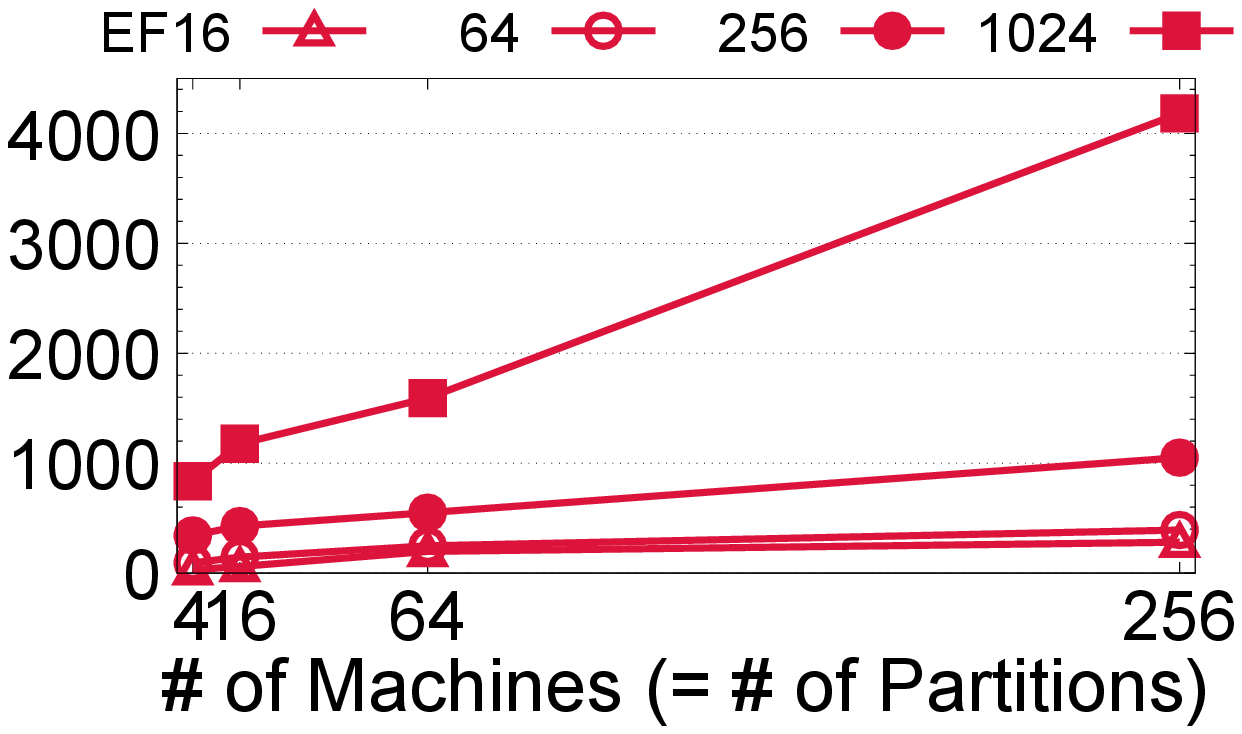}\label{fig:trillion}}
 \vspace{-13pt}
 \caption{Elapsed Time (sec) to Partition Real-world Graphs and RMAT Graphs.}\label{fig:performance}
 \vspace{-15pt}
\end{figure*}


Figure~\ref{fig:diff-edge} shows the elapsed time to partition Scale22 into 64 partitions on the different edge factors (EFs).
As the edge factor increases exponentially, so does the elapsed time in all algorithms.
The rate of the increase in \textit{Distributed NE} is lower than the others. 
As a result, \textit{Distributed NE} is slower than \textit{XtraPuLP} in  16 and 64, but it becomes faster in  256 and 1024.

\newpage
Figure~\ref{fig:diff-vert} shows the elapsed time of RMAT with Edge factor 1024 on the different scales using 64 machines.
The result is straightforward.
The elapsed time increases as the scale of RMAT graph.
Their increasing rates are similar in all algorithms.

\subsection{Scalability to Trillion-edge Graph}
Figure~\ref{fig:trillion} shows the scalability of \textit{Distributed NE} to trillion-edge graph.
We fix the number of vertices per machine as $2^{22}$ and change the number of machines (= the number of partitions).
Namely, we use Scale24 on 4 machines, Scale26 on 16 machines, Scale28 on 64 machines, and Scale30 on 256 machines, respectively.
As the number of machines increases, so does the elapsed time linearly.
This is mainly because workload imbalance occurs in the vertex selection in expansion processes.
Since the expansion rate is completely different in each partition, few partitions which have much more boundary vertices become the bottleneck of the entire processing.
For example in Edge factor 1024 (EF 1024), the elapsed time of the vertex selection on 4 machines is less than 1 \% of the entire processing, but it becomes 30.3 \% on 256 machines.
The communication cost also linearly increases as the increase of the number of machines.

Finally, due to the space-efficient design and implementation, \textit{Distributed NE} can cope with the trillion-edge graph (Scale30, EF 1024) using only 256 machines.
Its elapsed time is 69.7 minutes.


\subsection{Comparison with Sequential Algorithms}
Table~\ref{tbl:sequential} shows the comparison with the state-of-the-art sequential and streaming algorithms in the middle-scale real-world graphs.
\textit{NE} is the sequential offline algorithm, whereas \textit{HDRF} and \textit{SNE} are the sequential streaming algorithms.
Each graph is partitioned into 64 subgraphs. 
Thus, \textit{Distributed NE} are computed on 64 machines.
Although \textit{NE} provides the best replication factor, \textit{Distributed NE} is much faster than the sequential algorithms.
 \begin{table}[h]
 \vspace{-15pt}
 \centering
 \caption{Comparison with Sequential Algorithms.}\label{tbl:sequential}
 \scalebox{0.75}{
 \begin{tabular}{|c|l|r|r|r|r|}
 \hline
     \multicolumn{2}{|c|}{} & \multicolumn{1}{c|}{\texttt{Pokec}} & \multicolumn{1}{c|}{\texttt{Flickr}} & \multicolumn{1}{c|}{\texttt{LiveJ.}} & \multicolumn{1}{c|}{\texttt{Orkut}} \\ \hlinewd{2\arrayrulewidth}
 \parbox[t]{2mm}{\multirow{4}{*}{\rotatebox[origin=c]{90}{RF}}} 
     & \textit{HDRF} \cite{Petroni:2015:HSP:2806416.2806424} & 6.92 & 3.33 & 4.71 & 10.42\\
     & \textit{NE} \cite{Zhang:2017:GEP:3097983.3098033}  & \textbf{2.71} & \textbf{1.51} & \textbf{1.72} & \textbf{3.05} \\
     & \textit{SNE} \cite{Zhang:2017:GEP:3097983.3098033} & 3.89 & 1.78 & 2.12 & 5.66\\
     & \textbf{\textit{Distributed NE}} & \textit{3.92} & \textit{1.72} & \textit{2.19} & \textit{4.60} \\ \hline
\parbox[t]{2mm}{\multirow{4}{*}{\rotatebox[origin=c]{90}{\small{Time (sec)}}}} 
     & \textit{HDRF} & 24.310 &  24.370 &  57.228 &  92.479 \\
     & \textit{NE}   & 61.890 &  62.910 & 143.690 & 182.288 \\
     & \textit{SNE}  & 82.999 & 131.926 & 370.335 & 206.482 \\
     & \textbf{\textit{Distributed NE}} & \textbf{\textit{1.029}} & \textbf{\textit{7.523}} & \textbf{\textit{3.309}} & \textbf{\textit{3.224}} \\ 
\hline
 \end{tabular}}
  \vspace{-5pt}
 \end{table}


\begin{highlight}
\begin{table*}[t]
\newcolumntype{B}{!{\vrule width 3\arrayrulewidth}}
\caption{Performance of Graph Applications (\textit{SSSP}, \textit{WCC}, \textit{PageRank}) on 64 Partitions. \texttt{EB} is Edge Balance. \texttt{VB} is Vertex Balance. \texttt{ET} is Elapsed Time (sec). \texttt{COM} is COMmunication cost (GB). \texttt{WB} is Workload Balance.}
\label{tbl:perf}
\scalebox{0.72}{
  \begin{tabular}{|l|lBr|r|rBr|r|rBr|r|rBr|r|rBr|r|rBr|r|rBr|r|r|} \hline
  \multicolumn{2}{|cB}{} & \multicolumn{3}{cB}{\texttt{Flickr}}  & \multicolumn{3}{cB}{\texttt{Pokec}} & \multicolumn{3}{cB}{\texttt{LiveJ.}} & \multicolumn{3}{cB}{\texttt{Orkut}} & \multicolumn{3}{cB}{\texttt{Twitter}}  & \multicolumn{3}{cB}{\texttt{FriendSter}} & \multicolumn{3}{c|}{\texttt{WebUK}} \\ \hline \hline
  \multicolumn{2}{|cB}{} & \texttt{RF} & \texttt{EB} & \texttt{VB} & \texttt{RF} & \texttt{EB} & \texttt{VB} & \texttt{RF} & \texttt{EB} & \texttt{VB} & \texttt{RF} & \texttt{EB} & \texttt{VB} & \texttt{RF} & \texttt{EB} & \texttt{VB} & \texttt{RF} & \texttt{EB} & \texttt{VB} & \texttt{RF} & \texttt{EB} & \texttt{VB} \\ \hline
  \multirow{5}{*}{\rotatebox[origin=c]{90}{Quality}} & \textit{Rand.} & 7.3 & 1.0 & 1.0 & 18.1  & 1.0 & 1.0 & 11.8 & 1.0 & 1.0 & 33.4 & 1.0 & 1.0 & 17.8 & 1.0 & 1.0 & 20.0 & 1.0 & 1.0 & 21.6 & 1.0 & 1.0 \\ 
                                                     & \textit{2D-R.}  & 4.4 & 1.0 & 1.0   & 9.1 & 1.0   & 1.0   & 6.8 & 1.0 & 1.0 & 12.7 & 1.0 & 1.0 & 9.1  & 1.0 & 1.0 & 8.3 & 1.0 & 1.0 & 10.1 & 1.0 & 1.0 \\ 
                                                     & \textit{Obli.} & 6.3 & 1.7 & 1.1   & 13.6 & 1.6  & 1.1   & 9.0 & 1.1 & 1.0 & 20.9 & 1.3 & 1.0 & 13.8 & 1.0 & 1.0 & 14.3 & 1.0 & 1.0 & 4.0 & 1.3 & 1.0 \\ 
                                                     & \textit{H.G.}  & 4.0 & 1.2 & 1.0 & 10.2 & 1.2  & 1.1   & 6.0 & 1.1 & 1.1 & 14.3 & 2.5 & 1.1 & 5.5 & 1.3 & 1.1 & 9.6 & 1.3 & 1.0 &  3.4  & 1.0 & 1.0 \\ 
                                            & \textbf{\textit{D.NE}}  & \textit{\textbf{1.8}} & \textit{1.1} & \textit{3.5} & \textit{\textbf{4.3}} & \textit{1.1} & \textit{1.2} & \textit{\textbf{2.5}} & \textit{1.1} & \textit{1.3} & \textit{\textbf{5.1}} & \textit{1.1} & \textit{1.6} & \textit{\textbf{2.9}} & \textit{1.1} & \textit{1.6} & \textit{\textbf{3.5}} & \textit{1.1} & \textit{1.9} & \textit{\textbf{1.5}} & 1.1 & 1.6 \\ \hline \hline
  \multicolumn{2}{|cB}{} & \texttt{ET} & \texttt{COM} & \texttt{WB} & \texttt{ET} & \texttt{COM} & \texttt{WB} & \texttt{ET} & \texttt{COM} & \texttt{WB} & \texttt{ET} & \texttt{COM} & \texttt{WB} & \texttt{ET} & \texttt{COM} & \texttt{WB} & \texttt{ET} & \texttt{COM} & \texttt{WB} & \texttt{ET} & \texttt{COM} & \texttt{WB} \\ \hline
  \multirow{5}{*}{\rotatebox[origin=c]{90}{\textit{SSSP}}} & \textit{Rand.} & 2.96  & 1.78 & 1.58 & 2.91 & 3.10 & 1.46 & 4.08 & 6.02 & 1.41 & 4.45 & 11.3 & 1.25 & 22.7 & 87 & 1.15 & 50.3 & 146 & 1.20 & 88.4 & 254 & 1.27 \\ 
                                                           & \textit{2D-R.}  & 2.98  & 1.16 & 1.36 & 2.63 & 1.70 & 1.32 & 3.36 & 3.70 & 1.16 & 3.25 &  5.2 & 1.22 & 14.0 & 53 & 1.22 & 27.3 &  73 & 1.27 & 60.6 & 141 & 1.21 \\ 
                                                           & \textit{Obli.} & 2.99  & 1.57 & 1.57 & 2.77 & 2.40 & 1.68 & 3.67 & 4.68 & 1.38 & 3.61 &  7.6 & 1.32 & 19.4 & 73 & 1.15 & 38.7 & 112 & 1.22 & 39.4 &  83 & 1.21 \\ 
                                                           & \textit{H.G.}  & 2.98  & 2.75 & 1.56 & 3.46 & 3.01 & 1.67  & 3.18 & 6.45 & 1.43 & 3.24 &  9.0 & 1.24 & 11.6 & 88 & 1.25 & 26.8 & 145 & 1.23 & N/A & N/A & N/A \\ 
                                          & \textbf{\textit{D.NE}} & \textbf{\textit{2.94}} & \textit{\textbf{0.63}} & \textit{1.28} & \textbf{\textit{2.63}} & \textbf{\textit{1.03}} & \textit{1.42} & \textbf{\textit{3.15}} & \textbf{\textit{1.83}} & \textit{1.46} & \textit{\textbf{2.48}} & \textbf{\textit{3.1}} & \textit{1.71} & \textbf{\textit{7.8}} & \textit{\textbf{30}} & \textit{1.34} & \textbf{\textit{17.6}} & \textbf{\textit{44}} & \textit{1.42} & \textbf{\textit{28.5}} & \textit{\textbf{58}} & \textit{1.43} \\ \hline
  \multirow{5}{*}{\rotatebox[origin=c]{90}{\textit{WCC}}} & \textit{Rand.}  & 4.77 & 3.87 & 1.30 & 6.58 & 8.33 & 1.30 & 10.08 & 14.7 & 1.25 & 17.50 & 31.1 & 1.16 & 89.3 & 156 & 1.18 & 286.0 & 406 & 1.12 & 396.2 & 733 & 1.16 \\ 
                                                          & \textit{2D-R.}   & 3.90 & 2.33 & 1.18 & 4.24 & 4.26 & 1.19 &  6.65 &  8.5 & 1.16 &  9.53 & 12.3 & 1.11 & 56.9 &  85 & 1.15 & 169.6 & 173 & 1.18 & 231.6 & 350 & 1.22 \\ 
                                                          & \textit{Obli.}  & 4.59 & 3.36 & 1.38 & 5.44 & 6.24 & 1.40 &  8.54 & 10.9 & 1.30 & 13.70 & 19.9 & 1.13 & 74.5 & 122 & 1.14 & 217.6 & 293 & 1.12 & 108.7 & 144 & 1.25 \\ 
                                                          & \textit{H.G.}   & 3.97 & 3.43 & 1.37 & 4.64 & 5.60 & 1.33 &  6.44 &  9.8 & 1.27 & 10.84 & 15.7 & 1.35 & 41.1 &  91 & 1.20 & 159.2 & 239 & 1.18 & 119 & 232 & 1.06 \\ 
                                          & \textbf{\textit{D.NE}} & \textit{\textbf{3.48}} &  \textit{\textbf{0.74}} & \textit{1.31} & \textbf{\textit{3.55}} & \textbf{\textit{1.94}} & \textit{1.30} & \textbf{\textit{4.69}} & \textit{\textbf{2.7}} & \textit{1.34} & \textit{\textbf{7.09}} & \textbf{\textit{5.2}} & \textit{1.24} & \textbf{\textit{31.1}} & \textbf{\textit{31}} & \textit{1.28} & \textbf{\textit{115.3}} & \textbf{\textit{71}} & \textit{1.26} & \textbf{\textit{61.2}} & \textbf{\textit{55}} & \textit{1.25} \\ \hline
  \multirow{5}{*}{\rotatebox[origin=c]{90}{\textit{PageRank}}} & \textit{Rand.} & 51.2 & 35.0 & 1.32 & 72.8 & 65.6 & 1.29 & 120.1 & 130 & 1.23 & 182.0 & 228 & 1.11 & 1568 & 1607 & 1.14 & 2820 & 2942 & 1.11 & 3370 & 3853 & 1.12 \\ 
                                                               & \textit{2D-R.}  & 36.2 & 19.8 & 1.14 & 45.4 & 32.6 & 1.13 &  79.1 & 71  & 1.13 & 93.2  & 91  & 1.05 & 863  &  798 & 1.11 & 1407 & 1239 & 1.07 & 1650 & 1826 & 1.09 \\ 
                                                               & \textit{Obli.} & 45.6 & 28.9 & 1.38 & 63.0 & 51.2 & 1.39 & 100.7 & 96  & 1.28 & 129.2 & 147 & 1.10 & 1223 & 1252 & 1.14 & 2070 & 2112 & 1.12 & 769  & 776 & 1.15 \\ 
                                                               & \textit{H.G.}  & 31.1 & 14.9 & 1.23 & 41.3 & 24.4 & 1.26 &  61.8 & 43  & 1.33 & 87.1  & 74  & 1.14 & 446  &  462 & 1.19 & 1253 & 1151 & 1.20 & 682  & 687 & 1.06 \\ 
                                              & \textbf{\textit{D.NE}} & \textit{\textbf{28.0}} & \textit{\textbf{4.6}} & \textit{1.69} & \textit{\textbf{34.4}} & \textbf{\textit{14.0}} & \textit{1.33} & \textit{\textbf{49.4}} & \textit{\textbf{20}} & \textit{1.36} & \textbf{\textit{65.4}} & \textbf{\textit{33}} & \textit{1.44} & \textbf{\textit{362}} & \textit{\textbf{216}} & \textit{1.35} & \textbf{\textit{806}} & \textbf{\textit{432}} & \textit{1.22} & \textit{\textbf{289}} & \textit{\textbf{137}} & \textit{1.36} \\ \hline
  \end{tabular}
}
\vspace{-10px}
\end{table*}

\subsection{Effect on Distributed Graph Applications}
We briefly evaluate the effect of the partitioning methods on the distributed graph applications with the different real-world graphs (see Table~\ref{tbl:real_world}).
We use 3 common graph applications which have different communication patterns: Single Source Shortest Path (\emph{SSSP}), Weakly Connected Component (\emph{WCC}) and PageRank (\emph{PR}). These are implemented on PowerLyra~\cite{Chen:2015:PDG:2741948.2741970} forked from PowerGraph~\cite{joseph2012powergraph}.
\emph{SSSP} is the lightest workload and only involves a few communications; \emph{WCC} is medium; and \emph{PR} is the heaviest, where all the vertices send messages to their destinations in every iteration.
We also refer the reader to the performance analysis of distributed graph applications, such as \cite{Han:2014:ECP:2732977.2732980,6877273,Verma:2017:ECP:3055540.3055543,abbas2018streaming}.
Our result is consistent with the results presented in these analyses.
In \emph{SSSP}, we select Vertex~0 as the source. 
In \emph{PR}, we conduct 100 times iterations.

Table~\ref{tbl:perf} shows the result with regard to partitioning quality, balanceness, elapsed time, and communication cost on 64 machines.
In these result, we ignore the initialization phase (e.g., system setup, data loading, and partitioning) to show only the application performance.
We run each application 5 times and show the median.
We compare \textit{Distributed NE} (\textit{D.NE}) with 4 methods in PowerLyra: \textit{Random} (\textit{Rand.}), \textit{2D-Radom}(\textit{2D-R.}), \textit{Oblivious}(\textit{Obli.}), and \textit{Hybrid Ginger} (\textit{H.G.}).
N/A means that we cannot run the program correctly due to the system error by PowerLyra.

The balanceness among partitions is generally defined as $B(\{x_{p}|\ p \in P\}) := \frac{\max x_{p}}{\bar{x}}$, where $\bar{x} := \frac{\sum x_{p}}{|P|}$.
Therefore, the edge balance (\texttt{EB}) is $B(\{|E_p|\})$; the vertex balance (\texttt{VB}) is $B(\{|V(E_p)|\})$; and the workload balance is $B(\{\mathit{LT}_p\})$, where $\mathit{LT}_p$ is the local elapsed time in Partition~$p$.

Overall, \textit{Distributed NE} outperforms the others in the elapsed time for all the cases due to the reduction of the communication cost.
Especially, the improvement of the elapsed time is significant in \emph{PR} due to its high-communication workload; on the contrary, that of \emph{SSSP} is few due to its low-communication workload.

\textit{Distributed NE} achieves the good edge balance due to the algorithmic constraint, whereas the vertex balance becomes worse in some cases (e.g., \texttt{Flickr}).
However, this does not seriously affect the elapsed time.
Although Distributed NE is not explicitly aware of the vertex balance, it is not significantly worse due to its algorithmic characteristics. That is, during computation, each $|V(E_p)|$ slightly increases at a similar speed because each part greedily selects the next expanding vertex so that the increase of $|V(E_p)|$ is as minimal as possible.

\subsection{Evaluation with Non-skewed Graphs}
We evaluate the effect of \textit{Distributed NE} to 3 real-world road networks (California: 1.96M vertices and 2.76M edges; Pennsylvania: 1.08M vertices and 1.54M edges; and Texas: 1.37M vertices and 1.92M edges~\cite{leskovec2009community}) as the representative example of the large-scale non-skewed graphs. 
\textit{Distributed NE} provides the similar or slightly better partitioning quality compared to the other methods.
However, our claim is that traditional vertex partitioning would be a good choice in some cases because \textit{Distributed NE} is basically developed for skewed graphs as we already discussed in Section~\ref{sec:intro}.

 \begin{table}[h]
 \vspace{-15pt}
 \centering
 \caption{Replication Factor of Road Networks~\protect\cite{leskovec2009community}.}\label{tbl:sequential}
 \scalebox{0.75}{
 \begin{tabular}{|l|r|r|r|r|r|r|r|r|}
 \hline
     & \textit{Rand.} & \textit{2D-R.} & \textit{Obli.} & \textit{H.G.} & \textit{P.M.} & \textit{Sheep}& \textit{X.P.} & \textbf{\textit{D.NE}}\\ \hlinewd{2\arrayrulewidth}
     Calif. & 3.72 & 3.54 & 2.13 & 2.32 & \textbf{1.002} & 1.03 & 1.12 & \textit{1.02} \\
     Penn. & 3.74 & 3.55 & 2.14 & 2.40 & \textbf{1.004} & 1.03 & 1.11 & \textit{1.01} \\
     Tex. & 3.70 &  3.51 & 2.13 & 2.35 & \textbf{1.003} & 1.03 & 1.12 & \textit{1.02} \\ \hline
 \end{tabular}}
  \vspace{-10pt}
 \end{table}
\end{highlight}

\section{Conclusion}\label{sec:conclusion}



In this paper, we presented Distributed NE, a parallel and distributed edge partitioning algorithm, which produces high-quality partitions of skewed graphs fast and at scale as well as provides the theoretical upper bound of the partitioning quality.
There are two directions to improve our proposal.
On one hand, the further speed-up technique will be established for coping with exascale graphs~\cite{Ueno2017}. 
On the other hand, the extension to more complicated graph structures, such as dynamic graphs~\cite{huang2016leopard}, and hypergraphs~\cite{kabiljo2017social}, will be investigated.
{\bf Distributed NE is publicly available in} \url{http://www.masahanai.jp/DistributedNE/}.

\section*{Acknowledgement}
This research was supported in part under Singapore Ministry of Education (MoE) Academic Research Fund, Tier 1 Grant, No: RG20-14; Guangdong Province Innovative and Entrepreneurial Team Programme, No: 2017ZT07X386; and Shenzhen Peackock Programme, No: Y01276105.
\balance
\bibliographystyle{abbrv}
\bibliography{ref}

\end{document}